\documentclass[a4paper, 12pt, titlepage, openany, final]{book}

\usepackage[english]{babel}
\usepackage[utf8]{inputenc}
\usepackage[T1]{fontenc}
\usepackage[toc, page]{appendix}
\usepackage{amsmath}
\usepackage{amssymb}
\usepackage{mathtools}
\usepackage{lmodern}
\usepackage{graphicx}
\usepackage{caption}
\usepackage{subcaption}
\usepackage{pdfpages}
\usepackage{multicol}
\usepackage{fancyhdr}
\usepackage{changepage}
\usepackage{setspace}
\usepackage{tocbibind}
\usepackage{enumitem}
\usepackage{amsthm}
\usepackage{accents}
\usepackage{csquotes}
\usepackage{lastpage}
\usepackage{geometry}

\usepackage{tgschola}
\usepackage{comment}
\usepackage{tikz-cd}

\def\mydate{\leavevmode\hbox{\the\year.\twodigits\month.\twodigits\day}}
\def\twodigits#1{\ifnum#1<10 0\fi\the#1}

\allowdisplaybreaks

\theoremstyle{definition}
\newtheorem{definition}{Definition}[section]
\newtheorem{theorem}{Theorem}[section]
\newtheorem{corollary}[theorem]{Corollary}

\makeatletter
\newcommand{\ostar}{\mathbin{\mathpalette\make@circled\star}}
\newcommand{\make@circled}[2]{%
	\ooalign{$\m@th#1\smallbigcirc{#1}$\cr\hidewidth$\m@th#1#2$\hidewidth\cr}%
}
\newcommand{\smallbigcirc}[1]{%
	\vcenter{\hbox{\scalebox{0.77778}{$\m@th#1\bigcirc$}}}%
}
\makeatother

\DeclareMathOperator{\Tr}{Tr}
\DeclareMathOperator{\Div}{div}
\DeclareMathOperator{\sgn}{sgn}
\DeclareMathOperator{\Ad}{Ad}
\DeclareMathOperator{\hor}{hor}

\DeclareMathOperator{\diag}{diag}

\newcommand{\w}[1]{\accentset{\bullet}{#1}}
\newcommand{\rlc}[1]{\accentset{\circ}{#1}}

\tikzcdset{every label/.append style = {font = \small}}

\usepackage{hyperref}
\hypersetup{
	colorlinks=true,
	linkcolor=black,
	citecolor=black,
}

\begin{document}
	\allowdisplaybreaks
	\setlength{\abovedisplayskip}{5pt}
	\setlength{\belowdisplayskip}{5pt}
	\abovedisplayshortskip
	\belowdisplayshortskip
	
	\setlength{\parskip}{0pt}
	\newgeometry{margin = 4.5cm}
	
	\thispagestyle{empty}
	
	\newgeometry{
		a4paper,
		inner=30mm,
		outer=25mm,
		bottom=25mm,
	}
	
\begin{center}
	{\large COMENIUS UNIVERSITY IN BRATISLAVA \\
		FACULTY OF MATHEMATICS, PHYSICS AND INFORMATICS}
\end{center}
%

\vspace{2cm}
\begin{figure}[!h]
	\centering
	\includegraphics[width=4cm]{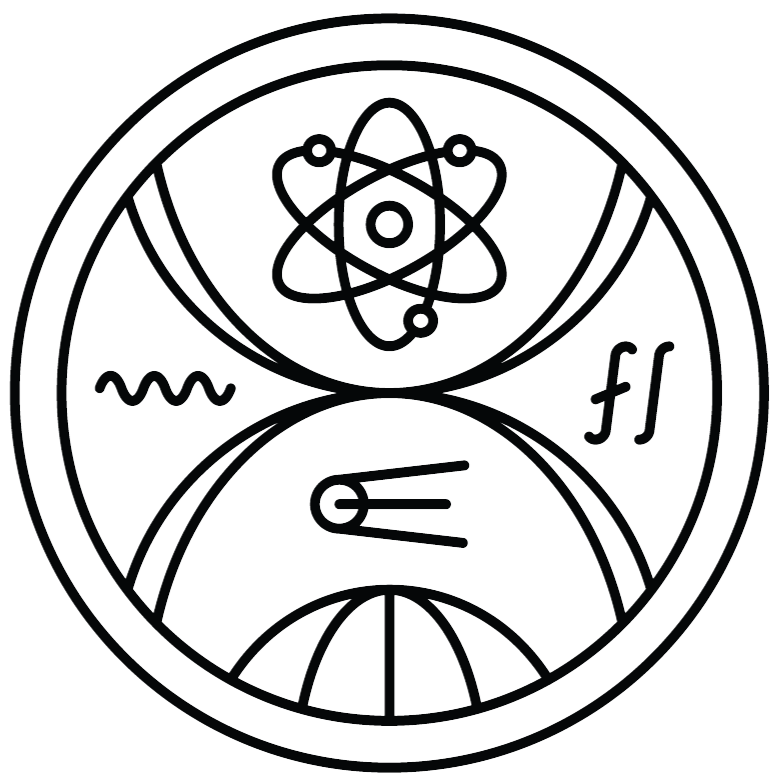}
\end{figure}

\vspace{1cm}
\begin{center}
	{\huge DUALITY OPERATOR IN TELEPARALLEL GRAVITY \\
		\vspace{3cm}
		\large MASTER'S THESIS}
\end{center}

\vfill
\begin{multicols}{2}
	{ \large
		\begin{flushleft} 2024 \end{flushleft}
		\begin{flushright} Bc. Marek Horňák \end{flushright} 
	}
\end{multicols}
\newpage
\thispagestyle{empty}
\,

\newpage
\thispagestyle{empty}
\begin{center}
	{\large COMENIUS UNIVERSITY IN BRATISLAVA \\
		FACULTY OF MATHEMATICS, PHYSICS AND INFORMATICS}
\end{center}

\vspace{5cm}
\begin{center}
	{\huge DUALITY OPERATOR IN TELEPARALLEL GRAVITY \\
		\vspace{3cm}
		\large MASTER'S THESIS}
\end{center}

\vfill
\begin{flushleft}
	\begin{tabular}{lll}
		Study program: & Theoretical physics \\
		Field: & Physics \\
		Department: & Department of Theoretical Physics \\
		Supervisor: & Mgr. Martin Krššák, Dr.rer.nat. \\
	\end{tabular}
\end{flushleft} 

\vfill
\begin{multicols}{2}
	\large
	\begin{flushleft} Bratislava 2024 \end{flushleft}
	\begin{flushright} {Bc. Marek Horňák} \end{flushright}
\end{multicols}





	
	\let\cleardoublepage\clearpage
\thispagestyle{empty}

\section*{Acknowledgements}
I would like to express my thanks to my supervisor, Mgr. Martin Krššák, Dr.rer.nat., for all the helpful discussions, insights, advice and all his patience throughout the journey. I am also grateful to my friends and family for their unwavering support and encouragement.  

\vspace*{\fill}
\section*{Declaration}
I declare that I have written this thesis on my own under the guidance of my supervisor and using the literature listed in the bibliography.

\vspace{1.5cm}
\begin{multicols}{2}
	\begin{flushleft} Bratislava 2024 \end{flushleft}
	\begin{flushright} 
		\par \makebox[3cm]{Bc. Marek Horňák}\makebox[1.5cm]{} \end{flushright}
\end{multicols}
\cleardoublepage
\frontmatter

\newgeometry{
	a4paper,
	inner=25mm,
	outer=30mm,
	bottom=25mm,
}

\newpage
\thispagestyle{plain}
\section*{Abstract}
\hspace{1cm}
\begin{tabular}{ll}
	\textit{Author}: & Bc. Marek Horňák \\
	\textit{Title}: & Duality Operator in Teleparallel Gravity \\
	\textit{University}: & Comenius University in Bratislava \\
	\textit{Faculty}: & Faculty of Mathematics, Physics and Informatics \\
	\textit{Department}: & Department of Theoretical Physics \\
	\textit{Supervisor}: & Mgr. Martin Krššák, Dr.rer.nat. \\
	\textit{City}: & Bratislava \\
	\textit{Date}: & \leavevmode\hbox{\the\day.\the\month.\the\year} \\
	\textit{Number of pages}: & \pageref{LastPage} \\
	\textit{Type of thesis}: & Master's thesis
\end{tabular}
\vspace{1cm}

Teleparallel gravity is a theory of gravity which replaces the Levi-Civita connection by a teleparallel connection -- a metric-compatible connection with vanishing curvature. 
Teleparallel equivalent of general relativity (TEGR) is a special case from this class of theories, and its dynamics is governed by equations which are equivalent to Einstein field equations of general relativity. 
In 2009 Lucas and Pereira presented the idea of a new operator that would allow us to express the action of TEGR in a form reminiscent of the Yang-Mills action for a gauge theory, for which torsion is the field strength.
The authors constructed the operator as a duality operator acting on torsion 2-forms and curvature 2-forms, taking into account all the contractions of these forms with the volume form of the spacetime manifold. 
The definition of this operator for other forms, however, remained unclear.  
In this thesis, we seek to give a mathematically consistent definition of the operator, and we follow the results published by Lucas and Pereira.
We show that the approach taken by Lucas and Pereira does not lead to a consistent definition of this operator, because it results in the operator with the desired properties either being non-unique or non-existent.


\begin{flushleft}
	\textbf{Keywords:} 
	
	duality operator, teleparallel gravity, tetrad formalism, torsion
\end{flushleft}

	\pagestyle{plain}
	\addtocontents{toc}{\protect\thispagestyle{plain}}
	\tableofcontents
	
	\let\cleardoublepage\clearpage
	\cleardoublepage
	\pagestyle{fancy}
	

	
	\mainmatter
	
	\chapter{Introduction}  

General relativity is an exceptionally successful theory of gravity. To this day, it stood several observational tests on almost all scales of physics and remains together with quantum field theory one of the pillars of modern physics. 

Albert Einstein formulated the field equations of general relativity in 1915. The fundamental idea of the theory is that gravity is a manifestation of the curvature of spacetime. It entirely disposed of the notion of gravitational force and described the motion of small bodies in the gravitational field as the natural motion of the bodies in the curved spacetime. The source of the curvature of spacetime is energy and momentum, as expressed by the Einstein field equations.  

In 1920s, Einstein -- who, of course, could not have been at that time aware of the need to modify his newly-established theory of gravity -- was attempting to develop a unified theory of gravity and electromagnetism based on the use of tetrads (orthonormal frames) and distant parallelism (teleparallelism). After publishing eight papers on the topic in the \textit{Sitzungsberichte} of the Prussian Academy, he abandoned his idea because it did not lead to his ultimate goal of unification of electromagnetism and gravity \cite{Sauer_2006}. Nevertheless, Einstein's work involved several important novelties, namely using tetrads instead of the metric tensor and considering other than Riemannian geometries to describe gravity. In 1960s and 1970s it became clear that Einstein teleparallelism can be a successful theory of gravity only, i.e. without attempts to incorporate electromagnetism \cite{Krssak:2024xeh}.

The development of quantum mechanics in 1920s and then quantum field theory in the following years started to expose some of the limitations of general relativity. Another challenge for general relativity came with the observations of the accelerating expansion of the universe \cite{SupernovaSearchTeam:1998fmf}. One possible approach to explain these phenomena is making modifications in the matter sector of the theory, which requires modifications of the Standard Model of particle physics. However, it may also be the case that it is not necessary to significantly restructure the Standard Model (the matter section), but that the gravitational section requires further revisiting. Thus, alternative theories of gravity are of considerable interest, because, instead of changing the matter content of the theory, an alternative theory of gravity may potentially provide an explanation for phenomena which within the current cosmological paradigm based on general relativity are ascribed to the presence of dark energy and dark matter \cite{Clifton:2011jh, Heisenberg:2018vsk}.

Our current understanding of elementary particles, which make up matter, and three of the four fundamental interactions of nature -- electromagnetic, strong and weak -- is captured by the Standard Model of particle physics. It is a quantum field theory in which all the interactions are mediated by so-called gauge bosons\footnote{The gauge bosons of Standard Model are photons (which mediate the electromagnetic interaction), eight gluons (for the strong interaction) and $W^+, W^-$ and $Z^0$ bosons (they mediate the weak interaction).}. Standard Model is an example of a gauge theory. 



Gauge theories have a special place among the field theories. One of the reasons is that they can be quantized. Gauge bosons are particles which correspond to gauge fields, which are fields that have to be introduced into a theory to make a global symmetry of this theory to be local. 
In other words, gauge theories have some unphysical degrees of freedom. The same physical fields correspond to several different configurations of the gauge fields. This is a kind of redundancy in the theoretical description of reality. 

Examples of gauge theories include the Standard Model and electrodynamics (classical, as well as quantum), which is a $U(1)$-gauge theory -- which means that it is invariant under $U(1)$ local transformations -- and the gauge bosons of this theory are photons. Of great importance for particle physics are Yang-Mills theories, which are non-abelian gauge theories. An example of a Yang-Mills theory is quantum chromodynamics (QCD), which is a $SU(3)$-gauge theory.

General relativity is coordinate-independent, which means that the laws are invariant under changes of coordinates. In other words, the equations take the same form in any coordinates. This, however, does not make it a gauge theory from the mathematical point of view, because conceptually this is a different kind of redundancy in the physical model. On the other hand, all the other fundamental forces of physics are described by gauge theories, so it is natural to expect that gravity too can be described in the framework of a gauge theory. 



There exist several modified theories of gravity, and also attempts to conceptualize general relativity as a gauge theory. An interesting possibility that has been gaining momentum in the recent decades is \textit{teleparallel gravity}. The interest in teleparallelism was revived in 1960s and 1970s by the works of M\o{}ller concerning the definition of the energy-momentum for gravity \cite{Moller1961,Moller1966,Moller1978}, and then by Pellegrini and Pleba\'{n}ski \cite{Pellegrini1963}. Another impulse for revisiting Einstein's teleparalellism came from the work of Hayashi and Nakano on gauge aspects of gravity \cite{Hayashi:1967se}. 

While in general relativity the gravitational interaction is realized via the curvature of the Levi-Civita connection, which is metric-compatible and torsion-free, in teleparallel gravity the interaction is an effect of torsion of a teleparallel connection, that is, a metric-compatible connection with zero curvature. Torsion in this case acts as a force, which similarly to the Lorentz force equation of electromagnetism appears as an effective force term on the right-hand side of the equation of motion of a freefalling particle \cite{Krssak:2018ywd}.

One of the many theories that this torsion-based approach has produced is the teleparallel equivalent of general relativity (TEGR), which is dynamically equivalent to general relativity, meaning that they cannot be distinguished through classical experiments. However, they are conceptually quite different and in scenarios such as quantum gravity the theories may exhibit different behaviors.

In 1967 Hayashi and Nakano published a paper \cite{Hayashi:1967se} which discussed teleparallel gravity as gauge theory of translations only, instead of gauging the Lorentz or Poincar\'{e} group. The Lagrangian density of the theory is quadratic in torsion, which is identified as the field strength of the translational group. Later, Cho has shown \cite{Cho:1976} that requiring covariance of the field equations restricts the form of the translational Lagrangian to the Lagrangian of TEGR, which Einstein considered already in 1929 \cite{Einstein1929a}. TEGR may thus be regarded as a gauge theory of translations. 

The action of TEGR may be rewritten in a form revealing a striking similarity with the Yang-Mills action for a gauge theory with field strength $\mathcal{F}$:
\begin{equation}
	\mathcal{S}_{YM} = - \int_{\mathcal{M}}{\Tr \left( \mathcal{F} \wedge \star \mathcal{F} \right)}.
\end{equation}
The Hodge $\star$ operator, however, has to be replaced by a different operator, which we shall denote $\ostar$. It turns out that if $\ostar$ is to be an operator, then $\ostar \ostar T = - T$, where $T$ is the torsion. This is true for the Hodge $\star$ operator too: $\star \star T = - T$. 

In the paper titled \enquote{Hodge Dual for Soldered Bundles} \cite{Lucas:2008gs}, Lucas and Pereira suggest that this $\ostar$ actually \textit{is} an operator\footnote{Lucas and Pereira denote it simply by $\star$, but we would like to make the distinction between this newly-proposed operator and the Hodge operator clear.}, and that its existence is a consequence of the presence of the solder form on the frame bundle. Their primary condition is that the operator be a duality operator, so that $\ostar \ostar = \pm 1$. However, the details of the definition of this operator are unclear. 

The aim of this thesis is to present the fundamental features of teleparallel gravity. Later, we explore the possibility of a proper and mathematically consistent definition of the $\ostar$ duality operator on soldered bundles presented in \cite{Lucas:2008gs, Lucas:thesis, Pereira}.

In chapter \ref{chapter2}, we start in a more general setting of metric-affine geometries. We derive the formula for decomposition of the coefficients of a general affine connection on a Riemannian manifold into the sum of three terms. 

Chapter \ref{chapter3} contains the description of the tetrad formalism and the definition of connection 1-forms, torsion 2-forms and curvature 2-forms. We also show that the general teleparallel spin connection may be expressed as $\w{\omega}^a_{\phantom{a} b \mu} = \left( \Lambda^{-1}\right)^a_{\phantom{a} c} \partial_{\mu} \Lambda^c_{\phantom{c} b}$ for some $\Lambda \in SO(1,3)$.

In chapter \ref{chapter4} we finally turn our attention to teleparallel gravity. We derive the action of TEGR, which leads to equations of motion. Then we rewrite the action in a form resembling the Yang-Mills action of a gauge theory. 

Chapter \ref{chapter5} explores the possibility of a proper definition of the new duality operator on soldered bundles introduced by Lucas and Pereira \cite{Lucas:2008gs}. Their construction of this operator is not entirely clear, so in this chapter we state the rules for its construction in an attempt to be consistent. These rules are motivated by expecting this new operator to be similar to the Hodge star operator and also by observing the results of Lucas and Pereira. We proceed to determine how the operator behaves when applied on torsion 2-forms, connection 1-forms and curvature 2-forms. 

The final chapter, chapter \ref{chapter6}, is dedicated to potential applications of this new duality operator for soldered bundles. We also give a general idea behind the gauge structure of TEGR. 

\subsection*{Notation}

The signature of a Lorentzian metric is $(-,+,+,+)$ and the Minkowski metric is $\eta = \diag \left( -1, +1, +1, +1 \right) $. Objects related to the Levi-Civita connection will by denoted by a circle over the symbol, e.g. $\rlc{\nabla}, \rlc{\Gamma}^{\alpha}_{\mu \nu}, \rlc{R}$, while objects related to the metric teleparallel connection will be denoted by a full circle over the symbol, e.g. $\w{\nabla}, \w{\Gamma}^{\alpha}_{\mu \nu}, \w{R}$. Throughout the entire thesis, we use units in which the speed of light is $c = 1$.

	\chapter{Metric-affine geometries} \label{chapter2}

The mathematical formalism of general relativity appears to only require a differentiable spacetime manifold with with one geometric object defined on it -- the metric tensor $g$. A differentiable manifold endowed with a metric tensor is called Riemannian manifold. However, a closer investigation reveals that there is another geometric structure needed in order to define covariant derivative, which is necessary to define the curvature tensor. This structure is called affine connection. In general relativity, the choice of a connection is special in the sense that it is completely dependent on the metric, and nothing else. In a more general geometric context, however, a connection on the manifold is independent of the metric. Even non-Riemannian manifolds may have a connection defined on them.

\section{Linear connection}

For the basics of differential geometry we follow here the exposition and notation from the textbook by Marián Fecko \cite{Fecko}.

\begin{definition}
	\textit{Affine connection} $\nabla$ on a differentiable manifold $\mathcal{M}$ is a mathematical structure such that to any vector field $W$ on $\mathcal{M}$ it assigns an operator $\nabla_W$, the covariant derivative in the direction of $W$, with the following properties:
	\begin{enumerate}
		\item it is a linear operator on the tensor algebra, and it preserves the type of the tensor
		\begin{align*}
			\nabla_W : \mathcal{T}^p_q \left( \mathcal{M} \right) \longrightarrow \mathcal{T}^p_q \left( \mathcal{M} \right),  && \\
			\nabla_W \left( A + \lambda B \right) = \nabla_W A + \lambda \nabla_W B, && A,B \in \mathcal{T}^p_q \left( \mathcal{M} \right), \lambda \in \mathbb{R}.
		\end{align*}
		\item its behavior on the tensor product obeys Leibniz rule
		\begin{align*}
			\nabla_W \left( A \otimes B\right) = \left( \nabla_W A \right) \otimes B + A \otimes \left( \nabla_W B \right), && A \in \mathcal{T}^p_q \left( \mathcal{M} \right), B \in \mathcal{T}^{p^{\prime}}_{q^{\prime}} \left( \mathcal{M} \right).
		\end{align*}
		\item on type $\binom{0}{0}$ tensors (functions) it gives
		\begin{align*}
			\nabla_W \psi = W \psi, && \psi \in \mathcal{F}\left( \mathcal{M}\right) \equiv \mathcal{T}^0_0 \left( \mathcal{M} \right).
		\end{align*}
		\item it commutes with contractions
		\begin{align*}
			\nabla_W \circ C = C \circ \nabla_W,
		\end{align*}
		where $C$ is any contraction.
		\item it is $\mathcal{F}$-linear in $W$, i.e.
		\begin{align*}
			\nabla_{V + fW} = \nabla_V + f \nabla_W, && V, W \in \mathcal{T}^1_0 \left( \mathcal{M} \right), f \in \mathcal{F} \left( \mathcal{M} \right).
		\end{align*}
	\end{enumerate}
\end{definition}

Let $\mathcal{M}$ be a smooth manifold with a metric tensor $g$ and an affine connection $\nabla$ and let $\partial_{\mu}$ be a coordinate basis on $\mathcal{M}$. We denote $\Gamma^{\rho}_{\nu \mu}$ the Christoffel symbols (of the second kind) for the connection: $\nabla_{\partial_{\mu}} \partial_{\nu} = \Gamma^{\rho}_{\nu \mu} \partial_{\rho}$. The Christoffel symbols of the first kind are then defined by $\Gamma_{\rho \nu \mu} \coloneqq g_{\rho \lambda} \Gamma^{\lambda}_{\nu \mu}$.

\begin{definition}
	\textit{Curvature} $R$ is a tensor of type $\binom{1}{3}$, i.e. $R \in \mathcal{T}^1_3 \left( \mathcal{M} \right) $, defined by
	\begin{equation}
		R \left( \alpha , U, V, W \right) = \langle \alpha, \left(  \left[ \nabla_U , \nabla_V \right]  - \nabla_{\left[ U,V \right]} \right) W \rangle .
	\end{equation}
	The components of curvature with respect to the basis $\partial_{\mu}$ are
		\begin{equation}
			R^{\alpha}_{\phantom{\alpha} \beta \mu \nu} = \partial_{\mu} \Gamma^{\alpha}_{\beta \nu} - \partial_{\nu} \Gamma^{\alpha}_{\beta \mu} + \Gamma^{\alpha}_{\lambda \mu} \Gamma^{\lambda}_{\beta \nu} - \Gamma^{\alpha}_{\lambda \nu} \Gamma^{\lambda}_{\beta \mu}.
		\end{equation}
\end{definition}

\begin{definition}
	\textit{Torsion} $T$ is a tensor of type $\binom{1}{2}$, i.e. $T \in \mathcal{T}^1_2 \left( \mathcal{M} \right) $, defined by 
	\begin{equation}
		T \left( \alpha , U, V \right) =  \langle \alpha, \nabla_U V - \nabla_V U - \left[ U,V \right] \rangle .
	\end{equation}
	The components of torsion with respect to $\partial_{\mu}$ are
	\begin{equation} \label{torsion_coordinate}
		T^{\rho}_{\mu \nu} = \Gamma^{\rho}_{\nu \mu} - \Gamma^{\rho}_{\mu \nu}.
	\end{equation}
\end{definition}

\begin{definition}
	\textit{Non-metricity} $Q$ is a tensor of type $\binom{0}{3}$, i.e. $Q \in \mathcal{T}^0_3 \left( \mathcal{M} \right) $, defined by
	\begin{equation}
		Q \left( U, V, W \right) = \left( \nabla_U g \right) \left( V, W \right).
	\end{equation}
	The components of non-metricity with respect to $\partial_{\mu}$ are
	\begin{equation}
		Q_{\mu \nu \rho} = \nabla_{\mu} g_{\nu \rho} = \partial_{\mu} g_{\nu \rho} - \Gamma^{\lambda}_{\nu \mu} g_{\lambda \rho} - \Gamma^{\lambda}_{\rho \mu} g_{\nu \lambda}.
	\end{equation}
\end{definition}

\section{Special cases of metric-affine geometries}

Metric-affine geometries are concerned with the study of geometry on smooth manifolds equipped with a metric tensor and an affine connection. For a general metric-affine setting, any of the tensorial quantities defined in the previous section may be nonzero. It may also happen that all of them are nonvanishing. However, several special cases of metric-affine geometries are obtained by requiring that some of these tensors be zero \cite{Hehl:1994ue,Bahamonde:2021gfp}.

\begin{itemize}
	\item \textbf{Flat connection}: $R^{\alpha}_{\phantom{\alpha} \beta \mu \nu} \equiv 0$ \\
	If both torsion and non-metricity are present, this case is known as (general) teleparallel geometry. \cite{Adak:2023ymc}
	\item \textbf{Symmetric connection}: $T^{\rho}_{\phantom{\rho} \mu \nu} \equiv 0$ \\
	In this case, $\Gamma^{\rho}_{\nu \mu} = \Gamma^{\rho}_{\mu \nu}$. This case is considered, e.g., in Weyl geometry used in Weyl gravity. 
	\item \textbf{Metric-compatible connection}: $Q_{\mu \nu \rho} \equiv 0$ \\
	Geometry with metric-compatible connection is known as Riemann-Cartan geometry. Such connections are used in Poincaré gauge theory. \cite{Hehl:1994ue}
	\item \textbf{Metric teleparallel geometry}: $R^{\alpha}_{\phantom{\alpha} \beta \mu \nu} \equiv 0$, $Q_{\mu \nu \rho} \equiv 0$ \\
	We will be primarily concerned with this case. From now, we will refer to this case as \underline{\textit{teleparallel geometry}}. \cite{Pereira}
	\item \textbf{Symmetric teleparallelism}: $R^{\alpha}_{\phantom{\alpha} \beta \mu \nu} \equiv 0$, $T^{\rho}_{\phantom{\rho} \mu \nu} \equiv 0$ \cite{Nester:1998mp}
	\item \textbf{RLC connection} \footnote{\enquote{RLC} in the name of the connection stands for \enquote{Riemann-Levi-Civita}, but is also a pun used by Marián Fecko in his book \cite{Fecko}.}: $T^{\rho}_{\phantom{\rho} \mu \nu} \equiv 0$, $Q_{\mu \nu \rho} \equiv 0$ \\
	The most well-known class of geometries, which is used in the description of gravity in general relativity. We will denote the Christoffel symbols of the RLC connection by a circle: $\accentset{\circ}{\Gamma}^{\rho}_{\mu \nu}$.
	\item \textbf{ (Locally) Minkowski space}: $R^{\alpha}_{\phantom{\alpha} \beta \mu \nu} \equiv 0$ , $T^{\rho}_{\phantom{\rho} \mu \nu} \equiv 0$, $Q_{\mu \nu \rho} \equiv 0$ \\
	Note that here (as well as in the previous cases) only local geometry of the space is being discussed. Its global topology may be different from that of Minkowski space. Consider, for example \cite{Fecko}, the metric induced from embedding of torus $T^4$ in $\mathbb{R}^8$ with metric tensor 
	\begin{equation}
		g = - dx^1 \otimes dx^1 - dx^2 \otimes dx^2 + dx^3 \otimes dx^3 + \dots + dx^8 \otimes dx^8.
	\end{equation}
	On $T^4$ we use coordinates $\tau, \alpha, \beta, \gamma$ so that
	\begin{align*}
		x^1 = \cos \tau, && x^2 = \sin \tau, && x^3 = \cos \alpha, && x^4 = \sin \alpha, \\
		x^5 = \cos \beta, && x^6 = \sin \beta, && x^7 = \cos \gamma, && x^8 = \sin \gamma.
	\end{align*}
	The metric induced on $T^4$ is then
	\begin{equation}
		g = - d\tau \otimes d\tau + d\alpha \otimes d\alpha + d\beta \otimes d\beta + d\gamma \otimes d\gamma.
	\end{equation}
	Obviously, the RLC connection of $T^4$ with this metric has zero curvature, while the global topology of $T^4$ is different from Minkowski space $\mathbb{R}^4$.
\end{itemize}

Throughout the text, we will denote the quantities related to the (metric) teleparallel connection by a full circle over the symbol, e.g. the teleparallel connection will be $\w{\nabla}$.

\section{Connection decomposition}

The Christoffel symbols of a RLC connection on $\mathcal{M}$ may be expressed in terms of the metric tensor $g_{\mu \nu}$ and its derivatives as
\begin{equation} \label{RLC}
	\accentset{\circ}{\Gamma}^{\rho}_{\mu \nu} = \dfrac{1}{2} g^{\rho \lambda} \left( g_{\lambda \mu , \nu} + g_{\lambda \nu , \mu} - g_{\mu \nu , \lambda} \right) .
\end{equation}
We may expect that for a general connection we will obtain a similar result, but now containing terms that depend on torsion and non-metricity. We start with the three equations
\begin{align*}
	Q_{\mu \nu \rho} &= g_{\rho \nu , \mu} - \Gamma^{\lambda}_{\nu \mu} g_{\lambda \rho} - \Gamma^{\lambda}_{\rho \mu} g_{\nu \lambda}, \\
	Q_{\nu \mu \rho} &= g_{\rho \mu , \nu} - \Gamma^{\lambda}_{\mu \nu} g_{\lambda \rho} - \Gamma^{\lambda}_{\rho \nu} g_{\mu \lambda}, \\
	- Q_{\rho \nu \mu} &= - g_{\mu \nu , \rho} + \Gamma^{\lambda}_{\nu \rho} g_{\lambda \mu} + \Gamma^{\lambda}_{\mu \rho} g_{\nu \lambda}.
\end{align*} 
Taking the sum and using (\ref{RLC}), we get
\begin{align}
	Q_{\mu \nu \rho} +	Q_{\nu \mu \rho} - Q_{\rho \nu \mu} &= 2 g_{\rho \lambda} \accentset{\circ}{\Gamma}^{\lambda}_{\mu \nu} - g_{\lambda \rho} \left( \Gamma^{\lambda}_{\mu \nu} + \Gamma^{\lambda}_{\nu \mu} \right) \\
	&+ g_{\lambda \mu} \left( \Gamma^{\lambda}_{\nu \rho} - \Gamma^{\lambda}_{\rho \nu} \right) + g_{\lambda \nu} \left( \Gamma^{\lambda}_{\mu \rho} - \Gamma^{\lambda}_{\rho \mu} \right),   
\end{align}
and finally, using (\ref{torsion_coordinate})
\begin{align}
	Q_{\mu \nu \rho} +	Q_{\nu \mu \rho} - Q_{\rho \nu \mu} &= 2 g_{\rho \lambda} \accentset{\circ}{\Gamma}^{\lambda}_{\mu \nu} - 2 g_{\lambda \rho} \Gamma^{\lambda}_{\mu \nu} - g_{\lambda \rho} T^{\lambda}_{\mu \nu} + g_{\lambda \mu} T^{\lambda}_{\rho \nu} + g_{\lambda \nu} T^{\lambda}_{\rho \mu}.   
\end{align}
Rearranging and raising the index $\rho$
\begin{align}
	\Gamma^{\rho}_{\mu \nu} &= \accentset{\circ}{\Gamma}^{\rho}_{\mu \nu} + \dfrac{1}{2} \left( - T^{\rho}_{\phantom{\rho} \mu \nu} + T^{\phantom{\mu} \rho}_{\mu \phantom{\rho} \nu} + T^{\phantom{\nu} \rho}_{\nu \phantom{\rho} \mu} \right) - \dfrac{1}{2}\left( Q_{\mu \nu}^{\phantom{\mu \nu} \rho} +	Q_{\nu \mu}^{\phantom{\nu \mu} \rho} - Q^{\rho}_{\phantom{\rho} \nu \mu}  \right).
\end{align}
We see that the connection (or more precisely, the Christoffel symbols of the connection) may be decomposed and expressed as \cite{Bahamonde:2021gfp}
\begin{equation} \label{con_decomposition}
	\Gamma^{\rho}_{\mu \nu} = \accentset{\circ}{\Gamma}^{\rho}_{\mu \nu} + K^{\rho}_{\phantom{\rho} \mu \nu} + L^{\rho}_{\phantom{\rho} \mu \nu},
\end{equation}
where we have introduced the \textit{\textbf{contortion}}
\begin{equation} \label{contorsion}
	K^{\rho}_{\phantom{\rho} \mu \nu} \coloneqq \dfrac{1}{2} \left( T^{\phantom{\mu} \rho}_{\mu \phantom{\rho} \nu} + T^{\phantom{\nu} \rho}_{\nu \phantom{\rho} \mu} - T^{\rho}_{\phantom{\rho} \mu \nu} \right) 
\end{equation}
and the \textit{\textbf{disformation}}
\begin{equation}
	L^{\rho}_{\phantom{\rho} \mu \nu} \coloneqq \dfrac{1}{2} \left( Q^{\rho}_{\phantom{\rho} \mu \nu} - Q^{\phantom{\mu} \rho}_{\mu \phantom{\rho} \nu} - Q^{\phantom{\nu} \rho}_{\nu \phantom{\rho} \mu} \right).
\end{equation}
Both are tensors of type $\binom{1}{2}$. Note that while disformation is symmetric, contortion is neither symmetric, nor antisymmetric in the lower indices. 

	\chapter{Tetrads} \label{chapter3}

General relativity is essentialy just geometry on a Riemannian manifold. In order to do any calculations on a manifold we need to use coordinates, although the results of such calculations remain independent of the choice of coordinates. In the standard formulation of general relativity, all tensors are expressed with respect to coordinate bases and Einstein field equations are also expressed using a coordinate basis. However, it is possible, and sometimes even convenient, to express them with respect to a non-coordinate basis.

\section{Holonomic and anholonomic frames}

Let $e_a$ be a frame defined in some region $\mathcal{O} \subset \mathcal{M}$ (i.e. a basis of $T_p \mathcal{M}$ in each $p \in \mathcal{O}$) and let $x^{\alpha}$ be coordinates in some region $\tilde{\mathcal{O}} \subset \mathcal{M}$ such that $\mathcal{O} \cap \tilde{\mathcal{O}} \neq \emptyset$. As $e_a$ and $\partial_{\alpha}$ are both bases, there exists a (non-singular) matrix $e_a^{\phantom{a} \mu} (x)$ such that
\begin{equation}
	e_a = e_a^{\phantom{a} \mu} (x) \partial_{\mu}.
\end{equation}
Similarly, we define the matrix $e^a_{\phantom{a} \mu} (x)$, which relates the corresponding dual bases:
\begin{equation}
	e^a = e^a_{\phantom{a} \mu} (x) dx^{\mu}.
\end{equation}
The condition of duality gives us
\begin{align}
	&e^a_{\phantom{b} \mu} e_b^{\phantom{b} \mu} = \delta^a_b ,	&e_a^{\phantom{a} \mu} e^a_{\phantom{a} \nu} = \delta^{\mu}_{\nu}.
\end{align}

Now we have two frames in a certain region: the coordinate frame $\partial_{\alpha}$ and the frame $e_a$. The components of a tensor $t$ with respect to the coordinate frame $\partial_{\alpha}$ are
\begin{equation} \label{t_tensor}
	t^{\alpha \dots \beta}_{\phantom{\alpha \dots \beta} \gamma \dots \delta} = t(dx^{\alpha}, \dots, dx^{\beta}; \partial_{\gamma}, \dots, \partial_{\delta}) ,
\end{equation}
while its components with respect to the frame $e_a$ are
\begin{equation}
	t^{a \dots b}_{\phantom{a \dots b} c \dots d} = t(e^{a}, \dots, e^{b}; e_{c}, \dots, e_{d}).
\end{equation}
We may wish to express the components with mixed indices, e.g.
\begin{equation}
	t^{\alpha \dots \beta a \dots b}_{\phantom{\alpha \dots \beta a \dots b} \gamma \dots \delta c \dots d} = t(dx^{\alpha}, \dots, dx^{\beta}, e^a, \dots, e^b; \partial_{\gamma}, \dots, \partial_{\delta}, e_c, \dots, e_d).
\end{equation}
Strictly speaking, these are not the components of the tensor (\ref{t_tensor}). The components of a tensor are defined by the tensor evaluated on vectors of one basis and covectors of the corresponding dual basis -- like in (\ref{t_tensor}) -- but here we are using two bases $\partial_{\alpha}$ and $e_a$.

It is clear that an upper Greek index may be changed to a Latin index by contracting with $e^a_{\phantom{a} \mu}$:
\begin{equation}
	t^{\dots a \dots}_{\phantom{\dots a \dots} \dots \dots} = e^a_{\phantom{a} \mu} t^{\dots \mu \dots}_{\phantom{\dots \mu \dots} \dots \dots},
\end{equation}
an upper Latin index may be changed to a Greek index by contracting with $e_a^{\phantom{a} \mu}$:
\begin{equation}
	t^{\dots \mu \dots}_{\phantom{\dots \mu \dots} \dots \dots} = e_a^{\phantom{a} \mu} \, t^{\dots a \dots}_{\phantom{\dots a \dots} \dots \dots},
\end{equation}
etc.

We define the \textit{\textbf{coefficients of anholonomy}} $c^a_{bc} (x)$ for a given frame $e_a$ by the relation
\begin{equation}
	\left[ e_b , e_c \right] = c^a_{bc} (x) e_a.
\end{equation}
For any coordinate frame,
\begin{equation}
	\left[ \partial_{\alpha} , \partial_{\beta} \right] = 0.
\end{equation}
That means that the coefficients of anholonomy for a coordinate frame vanish. Therefore, these frames are \textit{\textbf{holonomic}}.
The reverse statement is also true: if, for a given frame $e_a$, it holds that $\left[ e_a , e_b \right] = 0$, then $e_a$ is a coordinate frame, that is, in some neighbourhood of any point there exist local coordinates $x^{a}$ such that $e_a = \partial_a$. Frames for which $c^a_{bc} (x) \neq 0$ are called \textit{\textbf{anholonomic}} \cite{Fecko}.

Any two bases $E_a, \tilde{E}_b$ of the same 4-dimensional real vector space $L$ are related via some $A \in GL \left( 4,\mathbb{R} \right)$:
\begin{equation}
	E_a = A^b_{\phantom{b} a} \tilde{E}_b .
\end{equation} 
As any two frames $e_a, \tilde{e}_a$ on a 4-dimensional manifold $\mathcal{M}$ are at any given point $x \in \mathcal{M}$ just two bases of the tangent space $T_x \mathcal{M}$, which is a real vector space, the frames are related by local $GL\left( 4,\mathbb{R}\right)$ transformation:
\begin{equation}
	e_a (x) = A^b_{\phantom{b} a} (x) \tilde{e}_b (x).
\end{equation}

\section{Orthonormal frames}

Suppose that $\mathcal{M}$ is a Riemannian manifold with metric $g$. A special class of frames on $\mathcal{M}$ are the orthonormal frames. These frames are called vielbeins, and in the case $\dim \mathcal{M} = 4$ they are called vierbeins or \textit{\textbf{tetrads}} \cite{Carroll}. We will adopt the following convention: we will denote general frames by $e_a$, while tetrads by $h_a$. Authors commonly use the term tetrad to refer to any of $h_a, h^a, h^a_{\phantom{a} \mu}, h_a^{\phantom{a} \mu}$. 

By definition,
\begin{equation}
	g(h_a, h_b) = \eta_{ab},
\end{equation}
which means that
\begin{align} \label{metric_from_tetrad}
	&g = \eta_{ab} h^a \otimes h^b,	&g_{\mu \nu} = \eta_{ab} h^a_{\phantom{a} \mu} h^b_{\phantom{b} \nu}.
\end{align}
Of course, we can lower the Greek indices by contraction with $g_{\mu \nu}$ and raise them by contracting with $g^{\mu \nu}$. However, Latin indices are lowered by $\eta_{ab}$ and raised by $\eta^{ab}$.

An important corollary of (\ref{metric_from_tetrad}) is that we can change our viewpoint. We started this section by assuming that a metric $g$ is given on $\mathcal{M}$. Then we defined tetrads as orthonormal frames. However, we may also start from a tetrad $h_a$ on $\mathcal{M}$ and then use (\ref{metric_from_tetrad}) to construct the metric $g$ for which $h_a$ is a tetrad.

Let $h_a, \tilde{h}_a$ be two tetrads on $\mathcal{M}$. The tetrads are frames, so they are related by a local $GL\left( 4, \mathbb{R} \right)$ transformation
\begin{equation}
	h_a (x) = A^b_{\phantom{b} a} (x) \tilde{h}_b (x).
\end{equation}
However, $h_a$ and $\tilde{h}_b$ are tetrads, therefore
\begin{equation}
	\eta_{ab} =	g(h_a, h_b) = A^c_{\phantom{c} a} A^d_{\phantom{d} b} g(\tilde{h}_c, \tilde{h}_d) = \left( A^{T} \eta A \right)_{ab}.
\end{equation}
Thus, assuming that the signature of $g$ is $(1,3)$, two tetrads $h_a, \tilde{h}_a$ with the same orientation\footnote{From now on we will always assume that our tetrads are all right-handed.} are always related by a matrix $\Lambda(x) \in SO(1,3)$ \cite{Fecko}, i.e. by a local Lorentz transformation \cite{Carroll}.

\section{Connection, torsion and curvature forms}

Let $\mathcal{M}$ be a manifold, $\nabla$ a connection on $\mathcal{M}$ and let $e_a$ be a frame defined in some region $\mathcal{O} \subset \mathcal{M}$. The following definitions adhere to the exposition and notation from \cite{Fecko}.
\begin{definition} \label{connection_forms_def}
	The \textit{connection 1-forms} $\omega^a_{\phantom{a} b} \in \Omega^1 \left( \mathcal{M} \right)$ with respect to the frame $e_a$ are defined by
	\begin{equation}
		\omega^a_{\phantom{a} b}(V) e_a \coloneqq \nabla_V e_b.
	\end{equation}
\end{definition}
\begin{definition}
	The \textit{curvature 2-forms} $\Omega^a_{\phantom{a} b} \in \Omega^2 \left( \mathcal{M} \right)$ with respect to the frame $e_a$ are defined by
	\begin{equation}
		\Omega^a_{\phantom{a} b} \coloneqq \dfrac{1}{2} R^a_{\phantom{a} b \mu \nu} dx^{\mu} \wedge dx^{\nu}.
	\end{equation}
	We also define
	\begin{align}
		\omega_{ab} \coloneqq g_{ac} \omega^c_{\phantom{c} b} && \Omega_{ab} \coloneqq g_{ac} \Omega^c_{\phantom{a} b}.
	\end{align}
\end{definition} 
\begin{definition}
	The \textit{torsion 2-forms} $T^a \in \Omega^2 \left( \mathcal{M} \right)$ with respect to the frame $e_a$ are defined by
	\begin{equation}
		T^a \coloneqq \dfrac{1}{2} T^a_{\phantom{a} \mu \nu} dx^{\mu} \wedge dx^{\nu}.
	\end{equation}
\end{definition}

Further in this thesis, we will consider a connection which is compatible with a metric tensor of signature $(1,3)$ on a 4-dimensional manifold. In that case, if $\omega^a_{\phantom{a} b}$ are the 1- forms of this connection with respect to a tetrad (i.e. orthonormal frame) $h_a$, the fields $\omega^a_{\phantom{a} b \mu} (x)$ are commonly called \textit{\textbf{spin connection}} \cite{Fecko}. The connection 1-forms may be viewed as a single 1-form $\omega$ with values in the Lie algebra $\mathfrak{so}(1,3)$ of the Lorentz group.

\begin{theorem}
	Let $\nabla$ be a metric-compatible connection on $\mathcal{M}$, $\dim \mathcal{M} = 4$, for metric of signature $\left( 1,3 \right)$. The 1-forms of this connection $\omega^a_{\phantom{a} b}$ with respect to tetrad $h_a$ constitute a matrix from $\mathfrak{so}(1,3)$ algebra.
\end{theorem}

\begin{proof}

	Metric compatibility of a connection means that parallel transport preserves lengths and angles between vectors. Let us consider a vector field $V$, choose a point $p$ and consider the integral curve $\gamma(t)$ passing through it so that $\gamma(0) = p$, and let us parallel-transport vector $h_a$ of the tetrad from point $\gamma(\varepsilon)$ to point $p$ along $\gamma$. We denote the result of this parallel transport $h_a^{\parallel} (p)$. The covariant derivative is 
	\begin{equation}
		\left. \nabla_V h_a \right|_p = \lim_{\varepsilon \rightarrow 0} \dfrac{h_a^{\parallel} (p) - h_a (p)}{\varepsilon}.
	\end{equation}
	Because of metricity of the connection, $h_a^{\parallel} (p) = \Lambda^b_{\phantom{b} a} (\varepsilon) h_b (p)$, where $\Lambda \in SO(1,3)$ and $\Lambda (0) = \mathbb{I}$. Therefore by Definition \ref{connection_forms_def}
	\begin{equation}
		\left. \omega^b_{\phantom{b} a} (V) h_b \right|_p = \lim_{\varepsilon \rightarrow 0} \dfrac{\Lambda^b_{\phantom{b} a} (\varepsilon) - \delta^b_a }{\varepsilon} h_b (p),
	\end{equation}
	and so $\omega^b_{\phantom{b} a} (V)$ must be components of a matrix from the Lie algebra $\mathfrak{so}(1,3)$ of the Lorentz  group \cite{Fecko}.
\end{proof}

Cartan structure equations allow us to express the torsion 2-forms and curvature 2-forms with respect to frame (tetrad) $h_a$ using the coframe (cotetrad) $h^a$ and the connection 1-forms $\omega^a_{\phantom{a} b}$:
\begin{align}
	T^a &= dh^a + \omega^a_{\phantom{a} b} \wedge h^b,    \label{Cartan_struct_torsion} \\
	\Omega^a_{\phantom{a} b} &= d\omega^a_{\phantom{a} b} + \omega^a_{\phantom{a} c} \wedge \omega^c_{\phantom{c} b}.    \label{Cartan_struct_curvature}
\end{align}

\section{General teleparallel connection}

\begin{theorem}
	Let $\w{\nabla}$ be a metric teleparallel connection on $\mathcal{M}$, $\dim \mathcal{M} = 4$, for metric of signature $\left( 1,3 \right)$. There exists some (local) orthonormal frame (tetrad) $h_a$ on $\mathcal{M}$ such that the forms $\w{\omega}^a_{\phantom{a} b}$ of this connection with respect to $h_a$ vanish, that is 
	\begin{equation}
		\w{\omega}^a_{\phantom{a} b}=0.
	\end{equation}
\end{theorem}

\begin{proof}
	First, we prove an auxiliary statement \cite{Fecko}:
	
	Let $\omega$ be a connection 1-form on $O\mathcal{M}$ (the bundle of all orthonormal right-handed frames on $\mathcal{M}$) and $\Omega = D \omega = 0$ its curvature. Then there exists a (local) section $\sigma : \mathcal{M} \supset \mathcal{O} \longrightarrow O\mathcal{M}$ such that $\sigma^* \omega = 0$. \cite{Fecko}
	
	To prove this auxiliary statement, recall that $\omega$ is a 1-form whose component forms $\omega^i$ determine the horizontal distribution $\mathcal{D}$ on $O\mathcal{M}$, that is, a vector $U \in \mathcal{D}$ if and only if $\omega^i (U) = 0$. Further, by Frobenius theorem $\mathcal{D}$ is integrable if and only if $d\omega^i (U,V) = 0$ for every $U,V \in \mathcal{D}$. Equivalently, $\mathcal{D}$ is integrable if and only if $D\omega = \hor d \omega = 0$. That is because $\hor d \omega (U,V) = d \omega(\hor U,\hor V)$, and $\hor U$ and $\hor V$ are (by definition) horizontal vectors for any vectors $U,V$ on $O\mathcal{M}$. Now, because $\Omega \equiv D\omega = 0$, the horizontal distribution determined by $\omega$ is integrable. Consider an integral submanifold $\mathcal{S} \subset O\mathcal{M}$ given by distribution $\mathcal{D}$, and a section $\sigma : \mathcal{M} \supset \mathcal{O} \longrightarrow \mathcal{S}$. Then $\sigma$ is a horizontal section, which means that for any vector $W$ tangent to $\mathcal{M}$ we have $\sigma^* \omega (W) = \omega (\sigma_* W) = 0$ because $\sigma_* W$ is tangent to $\mathcal{S}$, hence horizontal. Thus $\sigma^* \omega = 0$.
	
	Now we use this auxiliary statement to prove our proposition. First, the metricity of $\w{\nabla}$ is necessary to ensure that the connection can be constructed as a connection on $O\mathcal{M}$. Second, we know that the curvature forms $\w{\Omega}^a_{\phantom{a} b} = 0$. However, $\w{\Omega}^a_{\phantom{a} b} \in \Omega^2 \left( \mathcal{M} \right)$, while in the auxiliary statement, $\Omega \in \Omega^2 (O\mathcal{M}, \Ad)$. We have to show that the curvature $\Omega \in \Omega^2 (O\mathcal{M}, \Ad)$ corresponding to $\w{\Omega}^a_{\phantom{a} b} \in \Omega^2 \left( \mathcal{M} \right)$ is also zero. We will use equation
	\begin{equation}
		\Omega = d\omega + \dfrac{1}{2} \left[ \omega \wedge \omega \right],
	\end{equation}
	where $\omega$ and $\Omega$ are matrix forms, i.e. $\omega^a_{\phantom{a} b}$ and $\Omega^a_{\phantom{a} b}$ are the components of $4 \times 4$ matrices, and we will treat them as such. Let us remark that $\left[ \omega \wedge \omega \right] = \omega^a \wedge \omega^b \left[ E_a, E_b \right] $, that is, the product of the component forms $\omega^a$ is the exterior product, and the basis vectors $E_a$ of the Lie algebra are treated as elements of algebra. Hence\footnote{$E^a_b$ is a matrix with components $\left( E^a_b \right)^c_d = \delta^a_d \delta^c_b$.}
	\begin{align*}
		\Omega^a_{\phantom{a} b} E^b_a &= d\omega^a_{\phantom{a} b} E^b_a + \dfrac{1}{2} \omega^a_{\phantom{a} c} \wedge \omega^d_{\phantom{d} b} \left[ E^c_a , E^b_d \right] \\
		&= d\omega^a_{\phantom{a} b} E^b_a + \dfrac{1}{2} \omega^a_{\phantom{a} c} \wedge \omega^d_{\phantom{d} b} \left( E^c_a E^b_d - E^b_d E^c_a \right) \\
		&= d\omega^a_{\phantom{a} b} E^b_a + \dfrac{1}{2} \omega^a_{\phantom{a} c} \wedge \omega^d_{\phantom{d} b} \left( \delta^c_d E^b_a - \delta^b_a E^c_d \right) \\
		&= d\omega^a_{\phantom{a} b} E^b_a + \dfrac{1}{2} \left( \omega^a_{\phantom{a} c} \wedge \omega^c_{\phantom{c} b} E^b_a - \omega^a_{\phantom{a} c} \wedge \omega^d_{\phantom{d} a} E^c_d \right) \\
		&= d\omega^a_{\phantom{a} b} E^b_a + \dfrac{1}{2} \left( \omega^a_{\phantom{a} c} \wedge \omega^c_{\phantom{c} b} - \omega^c_{\phantom{c} b} \wedge \omega^a_{\phantom{a} c} \right) E^b_a \\
		&= \left( d\omega^a_{\phantom{a} b} + \omega^a_{\phantom{a} c} \wedge \omega^c_{\phantom{c} b} \right) E^b_a.
	\end{align*}
	Applying pull-back $\sigma^*$ gives
	\begin{equation}
		\sigma^* \Omega^a_{\phantom{a} b} = d \left( \sigma^* \omega^a_{\phantom{a} b} \right) + \left( \sigma^* \omega^a_{\phantom{a} c} \right) \wedge \left( \sigma^* \omega^c_{\phantom{c} b} \right) = d \w{\omega}^a_{\phantom{a} b} + \w{\omega}^a_{\phantom{a} c} \wedge \w{\omega}^c_{\phantom{c} b} = \w{\Omega}^a_{\phantom{a} b}.
	\end{equation}
	Here $\w{\Omega}^a_{\phantom{a} b}$ are curvature 2-forms with respect to the tetrad $h_a$ that corresponds to the section $\sigma$. However, curvature is a tensor, so $\w{\Omega}^a_{\phantom{a} b} = 0$ with respect to any tetrad. Therefore $\sigma^* \Omega^a_{\phantom{a} b} = 0$ for any $\sigma$, which means $\Omega^a_{\phantom{a} b} = 0$. By the auxiliary statement, there exists a local section $\tilde{\sigma}$ such that $\w{\omega}^a_{\phantom{a} b} \equiv \tilde{\sigma}^* \omega^a_{\phantom{a} b} = 0$\footnote{This gauge is known as the \textit{Weitzenb\"{o}ck gauge} \cite{Bahamonde:2021gfp}.}. The section $\tilde{\sigma}$ corresponds to some tetrad $\tilde{h}_a$ on $\mathcal{M}$.
\end{proof}

\begin{corollary}
	Let $\w{\nabla}$ be a metric teleparallel connection on $\mathcal{M}$, $\dim \mathcal{M} = 4$, for metric of signature $\left( 1,3 \right)$, and let $\w{\omega}^a_{\phantom{a} b}$ be the 1-forms of this connection with respect to tetrad $h_a$. There exists a matrix 0-form $\Lambda \in \Omega^0 \left( \mathcal{M},SO(1,3) \right)$ such that \cite{Bahamonde:2021gfp}
	\begin{equation}
		\w{\omega}^a_{\phantom{a} b} = \left( \Lambda^{-1}\right)^a_{\phantom{a} c} d\Lambda^c_{\phantom{c} b}.
	\end{equation}
\end{corollary}

\begin{proof}
	Let $\w{\omega}^a_{\phantom{a} b}$ be the connection 1-forms with respect to tetrad $h_a$ and $\tilde{\w{\omega}}^a_{\phantom{a} b}$ the connection 1-forms with respect to tetrad $\tilde{h}_a$. The tetrads are related by a local Lorentz transformation
	\begin{equation}
		h_a (x) = \Lambda^b_{\phantom{b} a} (x) \tilde{h}_b,
	\end{equation}
	where $\Lambda (x) \in SO(1,3)$, and the connection 1-forms are transformed as \cite{Fecko}
	\begin{equation}
		\w{\omega}^a_{\phantom{a} b} = \left( \Lambda^{-1} \right)^a_{\phantom{a} c} \tilde{\w{\omega}}^c_{\phantom{c} d} \Lambda^d_{\phantom{d} b} + \left( \Lambda^{-1}\right)^a_{\phantom{a} c} d\Lambda^c_{\phantom{c} b}. 
	\end{equation}
	Now suppose that $\tilde{h}_a$ is a tetrad with respect to which $\tilde{\w{\omega}}^a_{\phantom{a} b} = 0$. Then 
	\begin{equation}
		\w{\omega}^a_{\phantom{a} b} = \left( \Lambda^{-1}\right)^a_{\phantom{a} c} d\Lambda^c_{\phantom{c} b},
	\end{equation}
	and we can view $\Lambda (x)$ as a matrix 0-form.
\end{proof}

We have thus found the connection 1-forms of a general metric teleparallel connection. Note that the the theorem and its corollary above may be restated for a manifold $\mathcal{M}$ of any dimension and a metric of any signature.

	\chapter{Teleparallel gravity} \label{chapter4}

\section{Action of teleparallel gravity}

The dynamics of general relativity is governed by the Einstein field equations, which can be obtained by varying action $\mathcal{S} = \rlc{\mathcal{S}} + \mathcal{S}_m$, where $\mathcal{S}_m$ is the matter action and $\rlc{\mathcal{S}}$ the Einstein-Hilbert action 
\begin{equation}
	\rlc{\mathcal{S}} = \frac{1}{2 \kappa} \int_{\mathcal{M}}{\rlc{R} \sqrt{-g} \, d^4 x}.
\end{equation}
We used the standard notation $g = \det (g_{\mu \nu})$, $\kappa = 8 \pi G$, where $G$ is the gravitational constant, and $\rlc{R}$ is the Ricci scalar, also called scalar curvature
\begin{equation}
	\rlc{R} \left( g_{\mu \nu} \right) \equiv \rlc{R}^{\alpha \beta}_{\phantom{\alpha \beta} \alpha \beta} 
\end{equation}
Variation of $\mathcal{S}$ leads to Einstein field equations (without the cosmological constant)
\begin{equation}
	\rlc{G}_{\mu \nu} = \kappa \Theta_{\mu \nu}.
\end{equation}
where $\rlc{G}_{\mu \nu} = \rlc{R}_{\mu \nu} - \frac{1}{2} \rlc{R} g_{\mu \nu}$ is the Einstein tensor, $\rlc{R}_{\mu \nu} = \rlc{R}^{\lambda}_{\phantom{\lambda} \mu \lambda \nu}$, $g_{\mu \nu}$ is the metric and $\Theta_{\mu \nu}$ is the energy-momentum tensor \cite{Carroll}.

To obtain an alternative mathematical description of gravity which is based on teleparallelism, we have to find the action for this theory \cite{Pereira}. The question is: What is our starting point? We know that general relativity is in good agreement with observations and experimental evidence. Therefore, our criterion for the construction of the Lagrangian of teleparallel gravity will be that the theory must agree with the results of general relativity. That means that the dynamical content of the equations of motion derived from the Lagrangian of teleparallel gravity must be the same as that of equations derived from the Einstein-Hilbert action, i.e. the Einstein field equations. 

We start \cite{Pereira, Bahamonde:2021gfp} by writing the the curvature tensor for a teleparallel connection\footnote{The result of this calculation is independent of the choice of frame, therefore we can do it in a holonomic frame $\partial_{\alpha}$.}:
\begin{equation}
	\w{R}^{\alpha}_{\phantom{\alpha} \beta \mu \nu} = \partial_{\mu} \w{\Gamma}^{\alpha}_{\beta \nu} - \partial_{\nu} \w{\Gamma}^{\alpha}_{\beta \mu} + \w{\Gamma}^{\alpha}_{\lambda \mu} \w{\Gamma}^{\lambda}_{\beta \nu} - \w{\Gamma}^{\alpha}_{\lambda \nu} \w{\Gamma}^{\lambda}_{\beta \mu}. 
\end{equation}
By (\ref{con_decomposition}), the (metric) teleparallel connection decomposes as
\begin{equation}
	\w{\Gamma}^{\rho}_{\mu \nu} = \rlc{\Gamma}^{\rho}_{\mu \nu} + \w{K}^{\rho}_{\phantom{\rho} \mu \nu} ,
\end{equation}
leading to 
\begin{equation}
	\begin{aligned}
		\w{R}^{\alpha}_{\phantom{\alpha} \beta \mu \nu} &= \rlc{R}^{\alpha}_{\phantom{\alpha} \beta \mu \nu} + 
		\partial_{\mu} \w{K}^{\alpha}_{\phantom{\alpha} \beta \nu} - \partial_{\nu} \w{K}^{\alpha}_{\phantom{\alpha} \beta \mu} + 
		\w{K}^{\alpha}_{\phantom{\alpha} \lambda \mu} \w{K}^{\lambda}_{\phantom{\lambda} \beta \nu} - \w{K}^{\alpha}_{\phantom{\alpha} \lambda \nu} \w{K}^{\lambda}_{\phantom{\lambda} \beta \mu} + \\
		& + \w{K}^{\alpha}_{\phantom{\alpha} \lambda \mu} \rlc{\Gamma}^{\lambda}_{\beta \nu} - \w{K}^{\alpha}_{\phantom{\alpha} \lambda \nu} \rlc{\Gamma}^{\lambda}_{\beta \mu} + 
		\rlc{\Gamma}^{\alpha}_{\lambda \mu} \w{K}^{\lambda}_{\phantom{\lambda} \beta \nu} - \rlc{\Gamma}^{\alpha}_{\lambda \nu} \w{K}^{\lambda}_{\phantom{\lambda} \beta \mu} \\
		&= \rlc{R}^{\alpha}_{\phantom{\alpha} \beta \mu \nu} + 
		\w{K}^{\alpha}_{\phantom{\alpha} \lambda \mu} \w{K}^{\lambda}_{\phantom{\lambda} \beta \nu} - \w{K}^{\alpha}_{\phantom{\alpha} \lambda \nu} \w{K}^{\lambda}_{\phantom{\lambda} \beta \mu} + \\
		& + \left( \partial_{\mu} \w{K}^{\alpha}_{\phantom{\alpha} \beta \nu} - \w{K}^{\alpha}_{\phantom{\alpha} \lambda \nu} \rlc{\Gamma}^{\lambda}_{\beta \mu} + \rlc{\Gamma}^{\alpha}_{\lambda \mu} \w{K}^{\lambda}_{\phantom{\lambda} \beta \nu} \right) 
		- \left( \partial_{\nu} \w{K}^{\alpha}_{\phantom{\alpha} \beta \mu} - \w{K}^{\alpha}_{\phantom{\alpha} \lambda \mu} \rlc{\Gamma}^{\lambda}_{\beta \nu} + \rlc{\Gamma}^{\alpha}_{\lambda \nu} \w{K}^{\lambda}_{\phantom{\lambda} \beta \mu} \right). 
	\end{aligned}
\end{equation}
Adding a zero term in the form $\w{K}^{\alpha}_{\phantom{\alpha} \beta \lambda} \rlc{\Gamma}^{\lambda}_{\mu \nu} - \w{K}^{\alpha}_{\phantom{\alpha} \beta \lambda} \rlc{\Gamma}^{\lambda}_{\nu \mu}$ allows us to simplify:
\begin{align*}
		\w{R}^{\alpha}_{\phantom{\alpha} \beta \mu \nu} &= \rlc{R}^{\alpha}_{\phantom{\alpha} \beta \mu \nu} + 
		\w{K}^{\alpha}_{\phantom{\alpha} \lambda \mu} \w{K}^{\lambda}_{\phantom{\lambda} \beta \nu} - \w{K}^{\alpha}_{\phantom{\alpha} \lambda \nu} \w{K}^{\lambda}_{\phantom{\lambda} \beta \mu} + \\
		& + \left( \partial_{\mu} \w{K}^{\alpha}_{\phantom{\alpha} \beta \nu} - \w{K}^{\alpha}_{\phantom{\alpha} \lambda \nu} \rlc{\Gamma}^{\lambda}_{\beta \mu} + \w{K}^{\lambda}_{\phantom{\lambda} \beta \nu} \rlc{\Gamma}^{\alpha}_{\lambda \mu} - \w{K}^{\alpha}_{\phantom{\alpha} \beta \lambda} \rlc{\Gamma}^{\lambda}_{\nu \mu} \right) \\
		& - \left( \partial_{\nu} \w{K}^{\alpha}_{\phantom{\alpha} \beta \mu} - \w{K}^{\alpha}_{\phantom{\alpha} \lambda \mu} \rlc{\Gamma}^{\lambda}_{\beta \nu} + \w{K}^{\lambda}_{\phantom{\lambda} \beta \mu} \rlc{\Gamma}^{\alpha}_{\lambda \nu} - \w{K}^{\alpha}_{\phantom{\alpha} \beta \lambda} \rlc{\Gamma}^{\lambda}_{\mu \nu} \right) \\
		&= \rlc{R}^{\alpha}_{\phantom{\alpha} \beta \mu \nu} + 
		\w{K}^{\alpha}_{\phantom{\alpha} \lambda \mu} \w{K}^{\lambda}_{\phantom{\lambda} \beta \nu} - \w{K}^{\alpha}_{\phantom{\alpha} \lambda \nu} \w{K}^{\lambda}_{\phantom{\lambda} \beta \mu} + \rlc{\nabla}_{\mu} \w{K}^{\alpha}_{\phantom{\alpha} \beta \nu} - \rlc{\nabla}_{\nu} \w{K}^{\alpha}_{\phantom{\alpha} \beta \mu} .
\end{align*}

The curvature of teleparallel connection vanishes, which implies that the scalar curvature
\begin{equation}
	0 = \rlc{R} + \w{K}^{\alpha}_{\phantom{\alpha} \lambda \alpha} \w{K}^{\lambda \beta}_{\phantom{\lambda \beta} \beta} - \w{K}^{\alpha}_{\phantom{\alpha} \lambda \beta} \w{K}^{\lambda \beta}_{\phantom{\lambda \beta} \alpha} + \rlc{\nabla}_{\alpha} \w{K}^{\alpha \beta}_{\phantom{\alpha \beta} \beta} - \rlc{\nabla}_{\beta} \w{K}^{\alpha \beta}_{\phantom{\alpha \beta} \alpha} . 
\end{equation}
Defining \textit{\textbf{torsion scalar}} \cite{Bahamonde:2021gfp}
\begin{equation} \label{torsion_scalar}
	\w{T} = \w{K}^{\alpha}_{\phantom{\alpha} \lambda \alpha} \w{K}^{\lambda \beta}_{\phantom{\lambda \beta} \beta} - \w{K}^{\alpha}_{\phantom{\alpha} \lambda \beta} \w{K}^{\lambda \beta}_{\phantom{\lambda \beta} \alpha} 
\end{equation}
and noticing
\begin{equation}
	\begin{aligned}
		\w{K}^{\alpha \beta}_{\phantom{\alpha \beta} \beta} - \w{K}^{\beta \alpha}_{\phantom{\alpha \beta} \beta} &= \dfrac{1}{2} \left( \w{T}^{\beta \alpha}_{\phantom{\beta \alpha} \beta} + \w{T}^{\phantom{\beta} \alpha \beta}_{\beta \phantom{\alpha \beta}} - \w{T}^{\alpha \beta}_{\phantom{\alpha \beta} \beta} - \w{T}^{\alpha \beta}_{\phantom{\beta \alpha} \beta} - \w{T}^{\phantom{\beta} \beta \alpha}_{\beta \phantom{\alpha \beta}} + \w{T}^{\beta \alpha}_{\phantom{\alpha \beta} \beta} \right) \\
		&= - 2 \w{T}^{\phantom{\beta} \beta \alpha}_{\beta \phantom{\alpha \beta}} ,
	\end{aligned}
\end{equation}
we can express the scalar curvature of the RLC connection as
\begin{equation}
	\rlc{R} = - \w{T} - \rlc{\nabla}_{\alpha} \left( \w{K}^{\alpha \beta}_{\phantom{\alpha \beta} \beta} - \w{K}^{\beta \alpha}_{\phantom{\beta \alpha} \beta} \right) = - \w{T} + 2 \rlc{\nabla}_{\alpha} \w{T}^{\phantom{\beta} \beta \alpha}_{\beta} .
\end{equation}
We will adopt the handy notation $h = \lvert \det (h^a_{\phantom{a} \mu}) \rvert$. Looking at (\ref{metric_from_tetrad}) we see that
\begin{equation*}
	\sqrt{-g} = \sqrt{- \det (g_{\mu \nu})} = \sqrt{- \det (\eta_{ab}) \det (h^a_{\phantom{a} \mu}) \det (h^b_{\phantom{a} \nu})} = \lvert \det (h^a_{\phantom{a} \mu}) \rvert = h. 
\end{equation*}
All this should serve as a good motivation for postulating the following:
\begin{definition}
	The \textit{action of teleparallel gravity} is
	\begin{equation}
		\w{\mathcal{S}} = - \dfrac{1}{2 \kappa}\int_{\mathcal{M}}{\w{T} h \, d^4 x}.
	\end{equation}
\end{definition}

To be more precise, we should call this the action of \textit{teleparallel equivalent of general relativity} (TEGR) \cite{Bahamonde:2021gfp}, but some authors (e.g. Aldrovandi and Pereira \cite{Pereira}) actually mean TEGR when they speak of teleparallel gravity.

The equations of motion may be obtained through the variation of this action. It is vitally important to realize that the dynamics of the theory with this action is by construction equivalent to the dynamics of general relativity, in which the dynamics is determined by the Einstein-Hilbert action. The reason is that $\left( 2 \kappa \right)^{-1} \int_{\mathcal{M}}{2 \rlc{\nabla}_{\alpha} \w{T}^{\phantom{\beta} \beta \alpha}_{\beta} \sqrt{-g} \, d^4 x}$ constitutes only a surface term, which is irrelevant when varying the action.
To show this, we use the formula \cite{Carroll}
\begin{equation}
	\rlc{\nabla}_{\mu} V^{\mu} = \dfrac{1}{\sqrt{-g}} \partial_{\mu} \left( \sqrt{-g} V^{\mu} \right),
\end{equation}
which allows us to write
\begin{equation}
	\dfrac{1}{2 \kappa}\int_{\mathcal{M}}{2 \rlc{\nabla}_{\alpha} \w{T}^{\phantom{\beta} \beta \alpha}_{\beta \phantom{\alpha \beta}} \sqrt{-g} \, d^4 x} = \dfrac{1}{2 \kappa}\int_{\mathcal{M}}{2 \partial_{\alpha} \left( \sqrt{-g} \w{T}^{\phantom{\beta} \beta \alpha}_{\beta \phantom{\alpha \beta}} \right) d^4 x} .
\end{equation}

Now let us rewrite the torsion scalar in terms of the torsion $\w{T}^{\lambda}_{\alpha \beta}$. Using (\ref{torsion_scalar}) and (\ref{contorsion}) we have
\begin{align*}
		\w{T} &= \w{T}^{\alpha}_{\phantom{\alpha} \alpha \lambda}  \w{T}^{\beta \lambda}_{\phantom{\mu \lambda} \beta} - \dfrac{1}{4} \left( \w{T}^{\phantom{\lambda} \alpha}_{\lambda \phantom{\alpha} \beta} + \w{T}^{\phantom{\beta} \alpha}_{\beta \phantom{\alpha} \lambda} - \w{T}^{\alpha}_{\phantom{\alpha} \lambda \beta} \right)  \left( \w{T}^{\beta \lambda}_{\phantom{\beta \lambda} \alpha} + \w{T}^{\phantom{\alpha} \lambda \beta}_{\alpha \phantom{\lambda \beta}} - \w{T}^{\lambda \beta}_{\phantom{\lambda \beta} \alpha} \right) \\
		&= \w{T}^{\alpha}_{\phantom{\alpha} \alpha \lambda}  \w{T}^{\beta \lambda}_{\phantom{\mu \lambda} \beta} - \dfrac{1}{4} \left( \w{T}_{\lambda \alpha \beta} + \w{T}_{\beta \alpha \lambda} - \w{T}_{\alpha \lambda \beta} \right)  \left( \w{T}^{\beta \lambda \alpha} + \w{T}^{\alpha \lambda \beta} - \w{T}^{\lambda \beta \alpha} \right) \\
		&= \w{T}^{\alpha}_{\phantom{\alpha} \alpha \lambda}  \w{T}^{\beta \lambda}_{\phantom{\mu \lambda} \beta} - \dfrac{1}{4} \w{T}_{\lambda \alpha \beta} \left( \w{T}^{\beta \lambda \alpha} + \w{T}^{\alpha \lambda \beta} - \w{T}^{\lambda \beta \alpha} \right) - \\
		&- \dfrac{1}{4} \w{T}_{\beta \alpha \lambda} \left( \w{T}^{\beta \lambda \alpha} + \w{T}^{\alpha \lambda \beta} - \w{T}^{\lambda \beta \alpha} \right) + \dfrac{1}{4} \w{T}_{\alpha \lambda \beta} \left( \w{T}^{\beta \lambda \alpha} + \w{T}^{\alpha \lambda \beta} - \w{T}^{\lambda \beta \alpha} \right).
\end{align*}
Notice that 
\begin{equation*}
	\w{T}_{\beta \alpha \lambda} \left( \w{T}^{\beta \lambda \alpha} + \w{T}^{\alpha \lambda \beta} - \w{T}^{\lambda \beta \alpha} \right) = \w{T}_{\lambda \alpha \beta} \left( \w{T}^{\lambda \beta \alpha} + \w{T}^{\alpha \beta \lambda} - \w{T}^{\beta \lambda \alpha} \right),
\end{equation*} so the first two terms with the parantheses cancel. What remains is
\begin{align}
	\w{T} &= \w{T}^{\alpha}_{\phantom{\alpha} \alpha \lambda}  \w{T}^{\beta \lambda}_{\phantom{\mu \lambda} \beta} + \dfrac{1}{4} \w{T}_{\alpha \lambda \beta} \left( \w{T}^{\beta \lambda \alpha} + \w{T}^{\alpha \lambda \beta} - \w{T}^{\lambda \beta \alpha} \right) \\
	&= \w{T}^{\alpha}_{\phantom{\alpha} \alpha \lambda}  \w{T}^{\beta \lambda}_{\phantom{\beta \lambda} \beta} + \dfrac{1}{4} \w{T}_{\lambda \alpha \beta} \left( \w{T}^{\beta \alpha \lambda} + \w{T}^{\lambda \alpha \beta} - \w{T}^{\alpha \beta \lambda} \right) \\
	&= \w{T}^{\alpha}_{\phantom{\alpha} \alpha \lambda}  \w{T}^{\beta \lambda}_{\phantom{\beta \lambda} \beta} + \dfrac{1}{4} \w{T}_{\lambda \alpha \beta} \left( \w{T}^{\beta \alpha \lambda} + \w{T}^{\lambda \alpha \beta} + \w{T}^{\beta \alpha \lambda} \right) \\
	&= \dfrac{1}{4} \w{T}_{\lambda \alpha \beta}  \w{T}^{\lambda \alpha \beta} + \dfrac{1}{2} \w{T}_{\lambda \alpha \beta} \w{T}^{\beta \alpha \lambda} - \w{T}^{\alpha}_{\phantom{\alpha} \lambda \alpha}  \w{T}^{\beta \lambda}_{\phantom{\mu \lambda} \beta} \\
	&= \dfrac{1}{4} \w{T}^{\lambda}_{\phantom{\lambda} \alpha \beta} \w{T}_{\lambda}^{\phantom{\lambda} \alpha \beta} + \dfrac{1}{2} \w{T}^{\lambda}_{\phantom{\lambda} \alpha \beta} \w{T}_{\phantom{\beta \alpha} \lambda}^{\beta \alpha} - \w{T}^{\alpha}_{\phantom{\alpha} \lambda \alpha}  \w{T}^{\beta \lambda}_{\phantom{\mu \lambda} \beta}.
\end{align}
The Lagrangian of teleparallel gravity is therefore
\begin{equation}
	\w{\mathcal{L}} = - \dfrac{h}{2 \kappa} \left( \dfrac{1}{4} \w{T}^{\lambda}_{\phantom{\lambda} \alpha \beta} \w{T}_{\lambda}^{\phantom{\lambda} \alpha \beta} + \dfrac{1}{2} \w{T}^{\lambda}_{\phantom{\lambda} \alpha \beta} \w{T}_{\phantom{\beta \alpha} \lambda}^{\beta \alpha} - \w{T}^{\alpha}_{\phantom{\alpha} \lambda \alpha}  \w{T}^{\beta \lambda}_{\phantom{\mu \lambda} \beta} \right).
\end{equation}

\section{Equations of motion}

We are now ready to derive the equations determining the dynamics in teleparallel gravity. Before we can even approach the variation of the action, we need to clarify what the variables of the Lagrangian are. As we have mentioned in the first chapter, the fundamental variables on manifold $\mathcal{M}$ are metric $g$ and connection $\w{\nabla}$, therefore the Lagrangian naturally depends on the tetrad $h^a_{\phantom{a} \mu}$ (from which we construct the metric by (\ref{metric_from_tetrad})) and the spin connection $\w{\omega}^a_{\phantom{a} b \mu}$. Writing the Euler-Lagrange equations
\begin{equation}
	\partial_{\sigma} \dfrac{\partial \w{\mathcal{L}}}{\partial \left( \partial_{\sigma} h^a_{\phantom{a} \rho} \right)} - \dfrac{\partial \w{\mathcal{L}}}{\partial h^a_{\phantom{a} \rho}} = 0
\end{equation}
for the variation of the Lagrangian given in the section above with respect to the tetrad, we arrive at the equations 
\begin{equation} \label{eqs_of_motion}
	\partial_{\sigma} \left( h \w{S}_a^{\phantom{a} \rho \sigma} \right) - \kappa h \w{\mathcal{J}}_a^{\phantom{a} \rho} = \kappa h \Theta_a^{\phantom{a} \rho},
\end{equation}
where $\Theta_a^{\phantom{a} \rho} = -\dfrac{1}{h} \dfrac{\delta \mathcal{L}_s}{\delta h^a_{\phantom{a} \rho}}$ is the matter energy-momentum tensor and $\mathcal{L}_s$ is the Lagrangian of a general source field \cite{Krssak:2018ywd}.
The other two objects appearing in eq. (\ref{eqs_of_motion}) are the so-called \textit{\textbf{superpotential}}
\begin{equation} \label{superpotential}
	\w{S}_a^{\phantom{a} \rho \sigma} = \dfrac{1}{2} \left( \w{T}_{\phantom{\sigma \rho} a}^{\sigma \rho} + \w{T}_{a}^{\phantom{a} \rho \sigma} - \w{T}_{\phantom{\rho \sigma} a}^{\rho \sigma} \right) - h_a^{\phantom{a} \sigma} \w{T}_{\phantom{\theta \rho} \theta}^{\theta \rho} + h_a^{\phantom{a} \rho} \w{T}_{\phantom{\theta \sigma} \theta}^{\theta \sigma}
\end{equation}
and the so-called \textit{\textbf{gauge current}}
\begin{equation}
	\w{\mathcal{J}}_a^{\phantom{a} \rho} = \dfrac{1}{\kappa} h_a^{\phantom{a} \lambda} \w{S}_c^{\phantom{c} \nu \rho} \w{T}^{c}_{\phantom{c} \nu \lambda} + \dfrac{h_a^{\phantom{a} \rho}}{h} \w{\mathcal{L}} + \dfrac{1}{\kappa} \w{\omega}^{c}_{\phantom{c} a \sigma} \w{S}_c^{\phantom{c} \rho \sigma}.
\end{equation}

Now half the work is done, but we have yet to execute the variation with respect to the other variable of the Lagrangian, the spin connection $\accentset{\bullet}{\omega}^a_{\phantom{a} b \mu}$. For this, we will resort to a dirty trick \cite{Krssak:2015lba, Golovnev:2017dox}. In the previous section, we have shown that 
\begin{equation}
	\w{\mathcal{L}} \left( h^a_{\phantom{a} \rho}, \w{\omega}^a_{\phantom{a} b \mu} \right) = \rlc{\mathcal{L}} \left( h^a_{\phantom{a} \rho} \right) - \dfrac{1}{\kappa} \partial_{\alpha} \left( h \w{T}^{\phantom{\beta} \beta \alpha}_{\beta \phantom{\alpha \beta}} \right),
		\end{equation}
where 
\begin{equation}
	\accentset{\circ}{\mathcal{L}} \left( h^a_{\phantom{a} \rho} \right) = \frac{1}{2 \kappa} h \rlc{R} \left( h^a_{\phantom{a} \rho} \right) \equiv \frac{1}{2 \kappa} \sqrt{-g} \rlc{R} \left( g_{\mu \nu} \right)
\end{equation} 
is the Einstein-Hilbert Lagrangian, which, as we know, does not depend on spin connection $\accentset{\bullet}{\omega}^a_{\phantom{a} b \mu}$, but only on the metric, and therefore (by (\ref{metric_from_tetrad})) on the tetrad $h^a_{\phantom{a} \rho}$. Due to the fact that the second term on the right is just a 4-divergence, it will only contribute a surface term in the action. However, surface terms are irrelevant when it comes to variation of the action, so variations of $\accentset{\bullet}{\mathcal{S}}$ with respect to $\accentset{\bullet}{\omega}^a_{\phantom{a} b \mu}$ are identically vanishing, hence they give no equation \cite{Krssak:2018ywd}. 


\section{Alternative form of action} \label{action__dual_rewrite}

Notice that the Lagrangian of teleparallel gravity may be rewritten in a different form \cite{Pereira}. First, we factor out $\w{T}^{\lambda}_{\phantom{\lambda} \mu \nu}$
\begin{equation}
	\begin{aligned}
		\w{\mathcal{L}} &= - \dfrac{h}{2 \kappa} \left( \dfrac{1}{4} \w{T}^{\lambda}_{\phantom{\lambda} \mu \nu} \w{T}_{\lambda}^{\phantom{\lambda} \mu \nu} + \dfrac{1}{2} \w{T}^{\lambda}_{\phantom{\lambda} \mu \nu} \w{T}_{\phantom{\nu \mu} \lambda}^{\nu \mu} - \w{T}^{\lambda}_{\phantom{\lambda} \mu \lambda} \w{T}_{\phantom{\nu \mu} \nu}^{\nu \mu} \right) \\
		&= - \dfrac{h}{2 \kappa} \w{T}^{\lambda}_{\phantom{\lambda} \mu \nu} \left( \dfrac{1}{4} \w{T}_{\lambda}^{\phantom{\lambda} \mu \nu} + \dfrac{1}{2} \w{T}_{\phantom{\nu \mu} \lambda}^{\nu \mu} - \delta^{\nu}_{\lambda} \w{T}_{\phantom{\rho \mu} \rho}^{\rho \mu} \right) \\
		&= - \dfrac{h}{2 \kappa} \eta_{ab} h^b_{\phantom{b} \lambda} \w{T}^a_{\phantom{a} \mu \nu} \left( \dfrac{1}{4} \w{T}^{\lambda \mu \nu} + \dfrac{1}{2} \w{T}^{\nu \mu \lambda} - g^{\nu \lambda} \w{T}_{\phantom{\rho \mu} \rho}^{\rho \mu} \right) \\
		& \equiv - \dfrac{h}{2 \kappa} \eta_{ab} h^b_{\phantom{b} \lambda} \w{T}^a_{\phantom{a} \mu \nu} \w{H}^{\mu \nu \lambda}.
	\end{aligned}
\end{equation}
Next, we use Kronecker deltas to change the indices on $\w{H}$ and use the fact that $\w{T}^{\lambda}_{\phantom{\lambda} \mu \nu} = \w{T}^{\lambda}_{\phantom{\lambda} \left[ \mu \nu \right] }$
\begin{equation}
	\begin{aligned}
		\w{\mathcal{L}} &= - \dfrac{h}{2 \kappa} \eta_{ab} h^b_{\phantom{b} \lambda} \w{T}^a_{\phantom{a} \mu \nu} \delta^{\mu}_{\alpha} \delta^{\nu}_{\beta} \w{H}^{\alpha \beta \lambda} \\
		&= - \dfrac{h}{2 \kappa} \eta_{ab} h^b_{\phantom{b} \lambda} \w{T}^a_{\phantom{a} \mu \nu} \delta^{\left[\mu \right. }_{\alpha} \delta^{ \left. \nu \right] }_{\beta} \w{H}^{\alpha \beta \lambda}. 
	\end{aligned}
\end{equation}
Using the identity $\varepsilon^{\mu \nu \gamma \sigma} \varepsilon_{\alpha \beta \gamma \sigma} = 4 \delta^{\left[\mu \right. }_{\alpha} \delta^{ \left. \nu \right] }_{\beta}$ \cite{Fecko} this further leads to
\begin{equation}
	\begin{aligned}
		\w{\mathcal{L}} &= - \dfrac{h}{2 \kappa} \eta_{ab} h^b_{\phantom{b} \lambda} \w{T}^a_{\phantom{a} \mu \nu} \dfrac{1}{4} \varepsilon^{\mu \nu \gamma \sigma} \varepsilon_{\alpha \beta \gamma \sigma} \w{H}^{\alpha \beta \lambda} \\
		&= - \dfrac{h}{2 \kappa} 6 \eta_{ab} h^b_{\phantom{b} \lambda} \w{T}^a_{\phantom{a} \mu \nu} \varepsilon_{\gamma \sigma \alpha \beta}  \w{H}^{\alpha \beta \lambda} \dfrac{1}{4!} \varepsilon^{\mu \nu \gamma \sigma} \\
		&= - \dfrac{h}{2 \kappa} \eta_{ab} h^b_{\phantom{b} \lambda} 6 \w{T}^a_{\phantom{a} \left[ \mu \nu \right. } \varepsilon_{\left. \gamma \sigma \right]  \alpha \beta} \w{H}^{\alpha \beta \lambda} \dfrac{1}{4!} \varepsilon^{\mu \nu \gamma \sigma}.
	\end{aligned}
\end{equation}
Interestingly, $\w{T}^a_{\phantom{a} \left[ \mu \nu \right. } \varepsilon_{\left. \gamma \sigma \right]  \alpha \beta} \w{H}^{\alpha \beta \lambda}$ would almost look like the exterior product of the torsion forms $\w{T}^a$ with some other form, were it not for the fact that at this point we do not know if $\varepsilon_{\gamma \sigma \alpha \beta} \w{H}^{\alpha \beta \lambda}$ is a form. In order to be more explicit, let us rewrite $6 = \dfrac{(2+2)!}{2!2!}$ and define 
\begin{equation} \label{ostar_T}
	\left( \ostar \w{T} \right)^b_{\gamma \sigma}  = h^b_{\phantom{b} \lambda} h \varepsilon_{\gamma \sigma \alpha \beta} \left( \dfrac{1}{4} \w{T}^{\lambda \alpha \beta} + \dfrac{1}{2} \w{T}^{\beta \alpha \lambda} - g^{\beta \lambda} \w{T}_{\phantom{\rho \alpha} \rho}^{\rho \alpha} \right).
\end{equation}
Then
\begin{equation}
	\begin{aligned}
		\w{\mathcal{L}} &= - \dfrac{1}{2 \kappa} \eta_{ab} \dfrac{(2+2)!}{2!2!} \w{T}^a_{\phantom{a} \left[ \mu \nu \right. } \varepsilon_{\left. \gamma \sigma \right]  \alpha \beta} h^b_{\phantom{b} \lambda} h \w{H}^{\alpha \beta \lambda} \dfrac{1}{4!} \varepsilon^{\mu \nu \gamma \sigma} \\
		& \equiv - \frac{1}{2 \kappa} \eta_{ab} \left( \w{T}^a \wedge \ostar \w{T}^b \right)_{\mu \nu \gamma \sigma } \dfrac{1}{4!} \varepsilon^{\mu \nu \gamma \sigma}.
	\end{aligned}
\end{equation}
We see that $\w{H}$ is a tensor, while $h \varepsilon_{\gamma \sigma \alpha \beta} = \omega_{\gamma \sigma \alpha \beta}$ is the (metric-compatible) volume form (see Definition \ref{vol_form}), and therefore $\ostar \w{T}^b$ is a contraction of two tensors, thus a tensor. Moreover, this tensor is antisymmetric (in the lower indices), so it is also a 2-form. Hence the use of the exterior product $\wedge$ is justified. 

As $\ostar \w{T}^b$ are 2-forms and the underlying manifold $\mathcal{M}$ is 4-dimensional, it must be that for any $a,b$ 
\begin{align}
	\left( \w{T}^a \wedge \ostar \w{T}^b \right) &\propto \omega, \\
	 \left( \w{T}^a \wedge \ostar \w{T}^b \right)_{\mu \nu \gamma \sigma} &\equiv \left( \w{T}^a, \w{T}^b \right)_g \omega_{\mu \nu \gamma \sigma}. 
\end{align}
Recalling that\footnote{$\sgn g$ is the sign of $\det \left( g_{\mu \nu}\right) $.}  $\omega^{\mu \nu \gamma \sigma} = \dfrac{\sgn g}{h} \varepsilon^{\mu \nu \gamma \sigma} = - \dfrac{1}{h} \varepsilon^{\mu \nu \gamma \sigma}$, we can now rewrite the action
\begin{equation}
	\begin{aligned}
		\w{\mathcal{S}} &= \int_{\mathcal{M}}{\w{\mathcal{L}} \, d^4 x} \\
		&= \int_{\mathcal{M}}{\frac{1}{2 \kappa} \eta_{ab} \left( \w{T}^a, \w{T}^b \right)_g \omega_{\mu \nu \gamma \sigma}  \dfrac{1}{4!} \omega^{\mu \nu \gamma \sigma} h \, d^4 x} \\
		&= \int_{\mathcal{M}}{\frac{1}{2 \kappa} \eta_{ab} \left( \w{T}^a, \w{T}^b \right)_g \omega} \\
		&= \int_{\mathcal{M}}{\frac{1}{2 \kappa} \eta_{ab} \left( \w{T}^a \wedge \ostar \accentset{\bullet}{T}^b \right)} \label{action_gauge_similar}.
	\end{aligned}
\end{equation}
The action of teleparallel gravity expressed in this form resembles the Yang-Mills action for a gauge theory for which torsion $\w{T}^a$ is the field strength \cite{Pereira}.

\chapter{New duality operator} \label{chapter5}

\section{Duality}

In the last section, we have introduced a new symbol $\ostar$, but it remains unclear what this symbol represents, and whether it represents a well-defined operator. 
The similarity of (\ref{action_gauge_similar}) with the Yang-Mills action for a gauge theory, in which $\ostar$ is replaced by $\star$, leads us to turn our attention to the Hodge operator.
Here we give a brief review of the Hodge star operator and its most important properties \cite{Fecko}.
\begin{definition}
	A \textit{volume form} $\omega$ on a $n$-dimensional vector space $L$ is any nonzero $n$-form in $L$. \\
	Let $L$ be, in addition, an oriented vector space with a metric tensor $g$, and let $h_a$ be its right-handed orthonormal basis. Then the \textit{metric-compatible volume form} is
	\begin{equation}
		\omega = h^1 \wedge \dots \wedge h^n.
	\end{equation}
\end{definition}
\begin{definition} \label{Hodge}
	Let $L$ be a $n$-dimensional vector space with metric $g$ and orientation $o$ and let $\omega$ be the metric-compatible volume form. \textit{Hodge star operator} is a map $\star \colon \Lambda^p L^* \longrightarrow \Lambda^{n-p} L^*$, $\alpha \mapsto \star \alpha$ defined by
	\begin{equation}
		(\star \alpha)_{a \dots b} \coloneqq \dfrac{1}{p!} \alpha^{c \dots d} \omega_{a \dots b c \dots d},
	\end{equation}
	where $\alpha^{c \dots d} \equiv g^{cr} \dots g^{ds} \alpha_{r \dots s}$.
\end{definition}

We see immediately that $\star$ is a linear operator on $\Lambda^p L^*$. Perhaps the most important property of the Hodge operator is
\begin{equation}
	\star \star = \sgn g (-1)^{p(n+1)},
\end{equation} 
i.e. it is a duality operator.

\begin{definition} \label{vol_form}
	A \textit{volume form} $\omega$ on an $n$-dimensional manifold $\mathcal{M}$ is any $n$-form which is nonzero everywhere on $\mathcal{M}$. \\
	Let $\mathcal{M}$ be, in addition, an orientable\footnote{From now on, we will always assume this is the case.} manifold with a metric tensor $g$, and let $h_a$ be a right-handed tetrad. Then the \textit{metric-compatible volume form} is
	\begin{equation}
		\omega = h^1 \wedge \dots \wedge h^n.
	\end{equation}
\end{definition}

Interestingly, it turns out that $\left( \ostar \ostar \w{T}\right) = - \w{T}$. To show this, we first have to specify what we mean by applying $\ostar$ on $\ostar \w{T}$, because so far we have only defined the action of $\ostar$ on $\w{T}$. This step, however, is rather trivial. Looking more closely at the definition of $\ostar \w{T}$ in (\ref{ostar_T}) it is obvious that, algebraically, there is no difference between $\w{T}$ and $\ostar \w{T}$ -- both are antisymmetric in the lower indices, and $\w{T}$ is not special in any way. Therefore, for the second application of $\ostar$ we use the same formula:
\begin{equation} \label{ostar_T}
	\left( \ostar \ostar \w{T} \right)^b_{\gamma \sigma}  = h^b_{\phantom{b} \lambda} h \varepsilon_{\gamma \sigma \alpha \beta} \left[ \dfrac{1}{4} \left( \ostar \w{T}\right)^{\lambda \alpha \beta} + \dfrac{1}{2} \left( \ostar \w{T}\right)^{\beta \alpha \lambda} - g^{\beta \lambda} \left( \ostar \w{T}\right)_{\phantom{\rho \alpha} \rho}^{\rho \alpha} \right].
\end{equation}
We will skip the tedious calculation, but we will approach the problem from the opposite side, as this will enable us to discover something more interesting along the way and the result will make it obvious that the statement above is indeed true. We attempt to construct a duality operator using all possible contractions of $\w{T}$ with the volume form $\omega$, following the suggestion of Lucas and Pereira \cite{Lucas:2008gs}. 

\section[Dual torsion]{Dual torsion $\ostar \w{T}$} \label{dual_torsion}
Now, let us be more specific. The $\star$-operator in Definition \ref{Hodge} includes all possible contractions of $\alpha_{c \dots d}$ with $\omega_{a \dots b c \dots d}$. However, torsion is a tensor with one upper index and two lower indices in which it is antisymmetric. Therefore, there exist possible contractions which would be omitted if we used Definition \ref{Hodge}: 
\begin{equation}
	\left( \star \w{T} \right)^{a}_{\alpha \beta} = \dfrac{1}{2} \omega_{\alpha \beta \mu \nu} \w{T}^{a \mu \nu}.
\end{equation}
If we choose to view torsion as a 2-form, we have to keep in mind that it is a vector-valued form and $\w{T}^a$ are the component 2-forms. Thus, the index $a$ is not a spacetime index, so the Hodge operator \enquote{ignores it}.

Before trying to find all the possible contractions, let us first formulate some rules, which will be also useful later. 

\subsection{Rules for possible contractions} \label{contraction_rules}

We have already mentioned that the connection 1-form $\w{\omega}$, as well as the curvature 2-form $\Omega$, is a Lie algebra-valued form. They are matrix forms, the matrices being $4 \times 4$, and the component forms are $\w{\omega}^a_{\phantom{a}b}$ and $\Omega^a_{\phantom{a} b}$. Torsion may be viewed as a 2-form $\w{T}$ with values in the Lie algebra of the group $\mathbb{T}_4$ of translations in 4 dimensions -- the component forms\footnote{Although we will only calculate $\ostar T$ and $\ostar{\omega}$ for the case of teleparallel connection, from the treatment it will be easily seen that the calculation is valid for a general metric-compatible connection.} are $\w{T}^a$. We will call indices $a,b$ the \textbf{\textit{algebraic indices}} of these forms, as they pertain to the algebra in which the forms take values. $\mathbb{R}$-valued forms are tensors, so the forms $\w{\omega}^a_{\phantom{a}b}$, $\Omega^a_{\phantom{a} b}$ and $\w{T}^a$ will have additional indices, which we will call \textbf{\textit{tensorial indices}}, e.g. in $\Omega^a_{\phantom{a} b \mu \nu}$ the algebraic indices are $a,b$, and the tensorial indices $\mu, \nu$ can be changed to $c,d$ using tetrad $h_c^{\phantom{c} \mu}$.

We will work in a tetrad basis, and all the components of forms will be expressed with respect to the same tetrad $h_a$. Suppose $\alpha$ is a $p$-form with $q$ algebraic indices. When determining all possible contractions of $\alpha$ with the volume form $\varepsilon_{abcd}$, we will abide by the following rules:
	\begin{enumerate}
		\item \underline{\textbf{\textit{$\left( \ostar \alpha \right)$ has $q$ algebraic indices and $(n-p)$ tensorial indices.}}} \\
			This requirement ensures that $\ostar$ is similar to $\star$, which behaves this way when applied on $p$-forms with $q$ algebraic indices (these algebraic indices have nothing to do with the manifold or coordinates).
		\item \underline{\textbf{\textit{The tensorial indices of $\left( \ostar \alpha \right)$ must come from the volume form $\varepsilon_{abcd}$.}}} \\
			The motivation behind this rule also lies in our desire for $\ostar$ to be similar to $\star$. This rule ensures that $\ostar \alpha$ will be a $(n-p)$-form.
		\item \underline{\textbf{\textit{The behavior of $\left( \ostar \alpha \right)$ in the $q$ algebraic indices is the same as the}}} \\
			\underline{\textbf{\textit{behavior of $\alpha$ in its $q$ algebraic indices.}}} \\
			For example, for the connection $\w{\omega}^a_{\phantom{a} b c}$ it holds that $\w{\omega}_{ab c} = - \w{\omega}_{ba c}$ (it is pseudo-antisymmetric in $a,b$), so $\left( \ostar \w{\omega} \right)^a_{\phantom{a} b cde}$ must also be pseudoantisymmetric in $a,b$.
	\end{enumerate}

\subsection{Duality condition}

All the possible contractions of torsion with the volume form according to the rules stated above are\footnote{The result of this calculation is independent of the choice of frame, therefore we can work in a holonomic frame $\partial_{\alpha}$. }:
$\omega_{\alpha \beta \mu \nu} \w{T}^{\gamma \mu \nu}, \omega_{\alpha \beta \mu \nu} \w{T}^{\mu \gamma \nu}$ and $\omega_{\alpha \beta \mu \nu} \w{T}^{\lambda \mu}_{\phantom{\lambda \mu} \lambda} g^{\nu \gamma}$.

We try to construct an operator $\ostar$ using these contractions:
\begin{equation}
	\left( \ostar \w{T} \right)^{\gamma}_{\phantom{\gamma} \alpha \beta} = \omega_{\alpha \beta \mu \nu} \left( a \w{T}^{\gamma \mu \nu} + b \w{T}^{\mu \gamma \nu} + c \w{T}^{\lambda \mu}_{\phantom{\lambda \mu} \lambda} g^{\nu \gamma} \right), 
\end{equation}
and we will require that $\left( \ostar \ostar \w{T} \right) = - \w{T}$. In Appendix \ref{append:torsion} we show that this condition leads us to two equations
\begin{align}
	4a^2 + 2ab + 2ac - 2bc &= + 1,    \\
	b^2 + 2ab - 2ac + 2bc &= 0.  \label{second}
\end{align}
Taking the sum of these equations, we obtain the following:
\begin{equation}
	+ 1 = 4a^2 + 4ab + b^2 = (2a + b)^2 ,
\end{equation}
hence
\begin{equation} \label{b_parameter}
	b = - 2a \pm 1 .
\end{equation}
Combining the last equation with eq. (\ref{second}) yields
\begin{equation}
	c = \pm \dfrac{b}{2\left( a - b \right) } = \dfrac{\mp 2a + 1}{6a \mp 2}.
\end{equation}
We see that there are infinitely many possible constructions of $\ostar$-operator. Choosing $a = \frac{1}{2}$ and using the upper sign in (\ref{b_parameter}) leads to $b = c = 0$, so that $\ostar = \star$ acts exactly the same way as the Hodge operator. Choosing $a = \frac{1}{4}$ and again using the upper sign in (\ref{b_parameter}) leads to $b = \frac{1}{2}, c = -1$, which is exactly\footnote{It is easy to see, that $\varepsilon_{\gamma \sigma \alpha \beta} \w{T}^{\beta \alpha \lambda} = - \varepsilon_{\gamma \sigma \beta \alpha} \w{T}^{\beta \alpha \lambda} = - \varepsilon_{\gamma \sigma \alpha \beta} \w{T}^{\alpha \beta \lambda} = \varepsilon_{\gamma \sigma \alpha \beta} \w{T}^{\alpha \lambda \beta}$.} the operator constructed in a different way in the previous section!

\subsection{Additional rule} \label{add_rule}

To make the construction of $\ostar \w{T}$ unique, we might apply an additional assumption. Looking at Definition \ref{Hodge}, we can see that the summation implied by the repeated indices in $(\star \alpha)_{a \dots b} =\frac{1}{p!} \alpha^{c \dots d} \omega_{a \dots b c \dots d}$ contains only one independent term, because in each term in the sum the indices $c,...,d$ may be rearranged in ascending order. For example, for a 3-form in 5D, 
\begin{alignat*}{2}
	(\star \alpha)_{13} &=\dfrac{1}{3!} \alpha^{cde} \omega_{13 cde} &&\\
	&= \dfrac{1}{3!} \left( \alpha^{245} \omega_{13 245} \right. &&+ \left. \alpha^{425} \omega_{13 425} + \alpha^{254} \omega_{13 254} \right. \\
	& &&\left. + \alpha^{452} \omega_{13 452} + \alpha^{524} \omega_{13 524} + \alpha^{542} \omega_{13 542} \right) \\
	&= \alpha^{245} \omega_{13 245}. &&
\end{alignat*}
Therefore, the factor $\frac{1}{p!}$ prevents double-counting.

If we apply a similar rule for $\ostar$, that is, we include each independent term the same number of times, we have to require $b = 2a$, because $\w{T}^{\gamma \mu \nu}$ is antisymmetric in $\mu, \nu$, while in general there is no relationship between $\w{T}^{\mu \gamma \nu}$ and $\w{T}^{\nu \gamma \mu}$. This rule does not tell us anything about $c$ and its relationship with $a$, because $\w{T}^{\lambda \mu}_{\phantom{\lambda \mu} \lambda} g^{\nu \gamma}$ is a different tensor from $\w{T}^{\gamma \mu \nu}$. Using this additional rule, we obtain 
\begin{align}
	a = \pm \frac{1}{4}, && b = \pm \frac{1}{2}, && c = \mp 1,
\end{align}
hence $\ostar \w{T}$ is now given uniquely, up to a sign. Notice that the solution with the upper sign corresponds to (\ref{ostar_T}), which is what we were after. 

\section[Dual connection]{Dual connection $\ostar \w{\omega}$} \label{dual_connection}

Now we would like to see whether there are any other objects on which we can apply our would-be operator $\ostar$. Until now, we have only found one such object, the torsion (or any 2-form with 1 algebraic index, i.e. vector-valued 2-form), which means that $\ostar$ does not yet deserve the noble title \enquote{operator}. We can start by looking at the case of connection $\w{\omega}$.

\subsection{Possible contractions and parametrizations}

The whole procedure will be very similar to our treatment of the case of torsion from the previous section. Connection $\w{\omega}^a_{\phantom{a} b \mu}$ has 3 indices, and $\w{\omega}_{a b \mu} = - \w{\omega}_{b a \mu}$. The result should be a 3-form, and the three indices must be inherited from the volume form, which has therefore only 1 index left for contractions. Hence, the only possible terms are $\w{\omega}^a_{\phantom{a} bc} \varepsilon^b_{\phantom{b} def}$, $\w{\omega}^a_{\phantom{a} bc} \varepsilon^c_{\phantom{c} def}$ and $\w{\omega}^{af}_{\phantom{af} f} \varepsilon_{bcde}$:
\begin{equation} \label{ostar_omega}
	\left( \ostar \w{\omega} \right)^a_{\phantom{a} bcde} \equiv a \left( \w{\omega}^a_{\phantom{a} fb} \varepsilon^f_{\phantom{f} cde} - \w{\omega}^{\phantom{bf} a}_{bf} \varepsilon^f_{\phantom{f} cde} \right) + b \w{\omega}^a_{\phantom{a} bf} \varepsilon^f_{\phantom{f} cde} + c \left( \w{\omega}^{af}_{\phantom{af} f} \varepsilon_{bcde} - \w{\omega}^{\phantom{b} f}_{b \phantom{f} f} \varepsilon^a_{\phantom{a} cde} \right).
\end{equation}
The terms in brackets are antisymmetrised so that pseudo-antisymmetry of $\w{\omega}^a_{\phantom{a} b \mu}$ is passed on to $\left( \ostar \w{\omega} \right)^a_{\phantom{a} b \mu \nu \lambda}$. 

Now, as we wish to apply $\ostar$ on $\ostar \w{\omega}$, and $\ostar \w{\omega} \equiv \w{\Xi}$ is a 3-form, we have to repeat the same exercise we did for $\w{\omega}^a_{\phantom{a} b \mu}$, but now for $\w{\Xi}^a_{\phantom{a} b \mu \nu \lambda}$. Three indices of the volume form may be contracted. The only possible terms are $\w{\Xi}^a_{\phantom{a} bcde} \varepsilon^{cde}_{\phantom{cde} f}$, $\w{\Xi}^a_{\phantom{a} bcde} \varepsilon^{bcd}_{\phantom{bcd} f}$, $\w{\Xi}^a_{\phantom{a} bcde} \varepsilon^{\phantom{a} bc}_{a \phantom{bc} f}$, then  $\w{\Xi}^{ae}_{\phantom{ae} cde} \varepsilon^{cd}_{\phantom{cd} fg}$, $\w{\Xi}^{a}_{\phantom{a} bcae} \varepsilon^{bc}_{\phantom{bc} fg}$, $\w{\Xi}^{a \phantom{b} c}_{\phantom{a} b \phantom{c} ce} \varepsilon^{\phantom{a} b}_{a \phantom{b} fg}$ (however, the last of these is zero), and finally $\w{\Xi}^{de}_{\phantom{de} cde} \varepsilon^{c}_{\phantom{c} fgh}$, $\w{\Xi}^{c \phantom{bc} d}_{\phantom{c} bc \phantom{d} d} \varepsilon^{b}_{\phantom{b} fgh}$ (but the last one is zero):

\begin{equation}
	\begin{aligned}
		\left( \ostar \w{\Xi} \right)^a_{\phantom{a} bf} &\equiv \alpha \w{\Xi}^a_{\phantom{a} bcde} \varepsilon^{cde}_{\phantom{cde} f} + \beta \left( \w{\Xi}^a_{\phantom{a} cdeb} - \w{\Xi}^{\phantom{bcde} a}_{bcde} \right) \varepsilon^{cde}_{\phantom{cde} f} + \gamma \w{\Xi}^{c \phantom{deb} a}_{\phantom{c} deb} \varepsilon^{\phantom{c} de}_{c \phantom{de} f} + \\
		&+ \mu \left( \w{\Xi}^{ae}_{\phantom{ae} cde} \varepsilon^{cd}_{\phantom{cd} bf} - \w{\Xi}^{\phantom{b} e}_{b \phantom{e} cde} \varepsilon^{cda}_{\phantom{cda} f} \right) + \nu \left( \w{\Xi}^{e \phantom{cde} a}_{\phantom{e} cde} \varepsilon^{cd}_{\phantom{cd} bf} - \w{\Xi}^{e}_{\phantom{e} cdeb} \varepsilon^{cda}_{\phantom{cda} f} \right) + \\
		&+ \lambda \w{\Xi}^{de}_{\phantom{de} cde} \varepsilon^{ca}_{\phantom{ca} bf}.
	\end{aligned}
\end{equation}

\subsection{Duality condition}

Now, we have too many constants -- nine, in fact -- and only two conditions:
\begin{align}
	\left( \ostar \ostar \w{\omega} \right) &= \w{\omega}, \\
	\left( \ostar \ostar \w{\Xi} \right) &= \w{\Xi}.
\end{align}
We cannot expect that these conditions determine the constants \textit{uniquely}. As the calculations would be tedious and prone to error if done by hand, we resort to using \textit{xTensor}, a \textit{Mathematica} package from a suite of packages called \textit{xAct}. 

In Appendix \ref{append:connection}, we obtain 10 equations from these conditions, but then we discover that only 6 are independent. We divide them into two triples of equations. The first triple contains no $c$-term
\begin{align}
	1 &= \left( 2a + b \right) \left( 6\alpha + 4\beta - 2\gamma \right), \label{qqq} \\
	0 &= a \left( 6 \alpha + 4 \beta - 2 \gamma - 2 \mu + \nu \right) + b \left( -2 \mu + \nu \right), \label{xxx} \\
	0 &= a \left( 4 \mu -2 \nu +2 \lambda \right) -b \lambda \label{zzz},
\end{align}
and the other triple is
\begin{align}
	0 &= 4a \beta + 2b \beta + c \left( -6 \alpha +10 \beta +4 \gamma \right), \\
	0 &= a \left( 6 \alpha + 2\beta - 2\mu \right) -2b \mu + c \left( - 2\mu - 2\nu \right), \\
	0 &= 4a \gamma +2b \gamma + c \left( 8 \beta +8 \gamma \right),
\end{align}
The very first equation tells us that we can not have $a=b=0$ (which means that the other two equations must be linearly dependent), and that $b \neq -2a$. For what follows, it is useful to notice that also $b \neq -a$: if $b =-a$ were true, then eq. (\ref{xxx}) would reduce to $0 = a \left( 6 \alpha + 4 \beta - 2 \gamma \right)$, which leads to either $a=b=0$, in contradiction with (\ref{qqq}), or $6 \alpha + 4 \beta - 2 \gamma = 0$, in contradiction with (\ref{qqq}).

Now, let us express everything in terms of $a, b$  and $c$. First, we take the first three independent equations and simplify them by denoting $A \equiv 6 \alpha + 4 \beta - 2 \gamma$ and $B \equiv -2\mu +\nu$:
\begin{align}
	1 &= \left( 2a + b \right) A,  \\
	0 &= a \left( A + B \right) + bB, \label{rrr} \\
	0 &= a \left( -2B +2 \lambda \right) -b \lambda. \label{www}
\end{align}
Multiplying (\ref{rrr}) by $\left( 2a+b \right)$ leads to 
\begin{equation}
	0 = a + B\left( a+b\right) \left( 2a+b \right),
\end{equation}
and multiplying (\ref{www}) by $\left( a+b\right) \left( 2a+b \right)$ gives
\begin{equation}
	0 = 2a^2 + \left( 2a - b\right) \left( a+b\right) \left( 2a+b \right) \lambda.
\end{equation}
Having $b=2a$ would lead to $a=0$, but in that case the solution would actually have to be $a=b=0$, which is impossible. Therefore, we have
\begin{align}
	A &= \dfrac{1}{ 2a + b }, \label{eq:A} \\
	B &= - \dfrac{a}{\left( a+b\right) \left( 2a+b \right)}, \\
	\lambda &= - \dfrac{2a^2}{\left( 2a - b\right) \left( a+b\right) \left( 2a+b \right)}.
\end{align}
Now we can attempt to express $\alpha, \beta, \gamma, \mu, \nu$ in terms of $a,b,c$ using equations
\begin{align}
	0 &= -6c \alpha + \left( 4a + 2b + 10c\right)\beta + 4c \gamma, \label{eq:1} \\
	0 &= 8c\beta + \left(  4a + 2b + 8c \right) \gamma, \label{eq:2} \\
	A &= 6 \alpha + 4 \beta - 2 \gamma, \label{eq:AA} \\
	3a\alpha + a\beta &= \left( a + b + c \right)\mu + c \nu, \\
	B &= -2\mu +\nu. \label{eq_mu_nu}
\end{align}
The first three equations do not contain $\mu, \nu$, so we can perform Gauss elimination to express $\alpha, \beta$ and $\gamma$, and then express $\mu, \nu$ from the last two equations\footnote{If the denominator in the expression for $\gamma$ is zero, the system has no solution.}:
\begin{align}
	\alpha &= \frac{1}{6\left( 2a + b \right)} - \dfrac{2}{3} \beta + \dfrac{1}{3} \gamma, \label{alpha} \\
	\beta &= - \dfrac{2a + b + 4c}{4c} \gamma, \label{beta} \\
	\gamma &= - \dfrac{2c}{(2a + b)^2 + 24c^2 + 11c (2a + b)} \cdot \dfrac{c}{2a + b},  \label{gamma}\\
	\mu &= \dfrac{a \left( 3\alpha + \beta \right) - Bc}{a + b + 3c}, \\
	\nu &= B + 2\dfrac{a \left( 3\alpha + \beta \right) - Bc}{a + b + 3c}.
\end{align}
This is a solution only in case $c\neq 0$ and $a+b+3c \neq 0$. 

In case $c \neq 0$ and $a+b+3c = 0$ we have $\alpha$, $\beta$ and $\gamma$ given by (\ref{alpha}), (\ref{beta}) and (\ref{gamma}). In this case, we get a condition
\begin{equation}
	a \left( \dfrac{1}{2\left( 2a + b \right) } - \beta + \gamma \right) = - \dfrac{ac}{\left( a + b \right) \left( 2a + b \right) },
\end{equation}
which means that we can not write the solution in terms of $a,b,c$, but we have to replace one of them by $\mu$ or $\nu$, and then express the other one from equation (\ref{eq_mu_nu}).

In case $c = 0$ the solution is 
\begin{align}
	\alpha &= \frac{1}{6\left( 2a + b \right)}, \\
	\beta &= 0, \\
	\gamma &= 0, \\
	\mu &= \dfrac{a}{2\left( a + b \right) \left( 2a + b \right)}, \\
	\nu &= 0.
\end{align}

\subsection{Using the additional rule}

The additional rule discussed in section \ref{add_rule} would give us
\begin{align}
	\beta &= (-1)^{\xi} 3 \alpha, \\
	\gamma &= 3 \alpha, \label{eq:gamma} \\
	\nu &= (-1)^{\zeta} 2 \mu,
\end{align}
where $\xi, \zeta \in \{0,1\}$. However, looking at eq. (\ref{eq:AA}) we see that we can not have any of $\alpha, \beta, \gamma$ equal to $0$, as in that case $\alpha = \beta = \gamma = 0$, while $A \neq 0$, which contradicts (\ref{eq:AA}). According to (\ref{eq:2}), $\beta \neq - \gamma$, otherwise we have
\begin{equation}
	0 = -8c\gamma + \left(  4a + 2b + 8c \right) \gamma = 2 \left(  2a + b \right) \gamma ,
\end{equation}
which is impossible. Thus $\beta = \gamma = 3 \alpha$. Then, subtracting (\ref{eq:2}) from (\ref{eq:1}) gives
\begin{equation}
	0 = -6c \alpha + 2c \beta - 4c \beta = -12 c \alpha,
\end{equation}
leading to $c = 0$. This reduces eq. (\ref{eq:2}) to 
\begin{equation}
	0 = 2 \left(  2a + b \right) \gamma,
\end{equation}
again, an impossibility. Therefore, imposing the additional rule leads to a system of equations which has no solution.

\subsection{Transformation of dual connection}

We have to address one more issue: Is this construction of the $\ostar$-operator as an operator acting on connection forms canonical? That is, we are interested in its behavior under local Lorentz transformations.

The problem is that not all contractions are canonical tensor operations, e.g the contraction implied by summation through $b$ in $\w{\omega}^a_{\phantom{a} bc} \varepsilon^b_{\phantom{b} def}$ is, strictly speaking, \textit{not} a contraction of two tensors. 

Changing the tetrad
\begin{equation}
	h_a \longrightarrow \tilde{h}_a = \Lambda^b_{\phantom{b} a} (x) h_b
\end{equation}
leads to the transformation \cite{Fecko}
\begin{align}
	\eta^{ab} &\longrightarrow \eta^{cd} \left( \Lambda^{-1} \right)^a_{\phantom{a} c} \left( \Lambda^{-1} \right)^b_{\phantom{b} d} = \eta^{ab}, \\
	\varepsilon_{abcd} &\longrightarrow \varepsilon_{efgh} \Lambda^e_{\phantom{e} a} \Lambda^f_{\phantom{f} b} \Lambda^g_{\phantom{g} c} \Lambda^h_{\phantom{h} d} = \left( \det \Lambda \right) \varepsilon_{abcd} = \varepsilon_{abcd}, \\
	\w{\omega}^a_{\phantom{a} bc}  &\longrightarrow \tilde{\w{\omega}}^a_{\phantom{a} bc} = \left( \Lambda^{-1} \right)^a_{\phantom{a} d} \w{\omega} ^d_{\phantom{d} ef} \Lambda^e_{\phantom{e} b} \Lambda^f_{\phantom{f} c} + \tilde{h}_c^{\phantom{c} \mu} \left( \Lambda^{-1} \right)^a_{\phantom{a} d} \partial_{\mu} \Lambda^d_{\phantom{d} b}. \label{connection_transformation}
\end{align}
The first term on the right in (\ref{connection_transformation}) is exactly what should be there if $\w{\omega}$ were a tensor of type $\binom{1}{2}$. Therefore, to make the following diagram
\[
\begin{tikzcd}[row sep=huge,column sep=huge]
	\w{\omega} \arrow{d}{\ostar} \arrow{r}{\Huge{\Lambda}} & \tilde{\w{\omega}} \arrow{d}{\ostar} \\
	\ostar{\w{\omega}} \arrow{r}{\Lambda} & \ostar{\tilde{\w{\omega}}}
\end{tikzcd}
\]
commutative, where $\Lambda$ is a local Lorentz transformation, the dual connection has to respond to the local Lorentz transformation in the following manner:
\begin{align*}
	\begin{split}
		\left( \ostar \w{\omega} \right)^a_{\phantom{a} bcde} \longrightarrow	& \left( \ostar \tilde{\w{\omega}} \right)^a_{\phantom{a} bcde} = a \left( \tilde{\w{\omega}}^a_{\phantom{a} fb} - \tilde{\w{\omega}}^{p}_{\phantom{p} fq} \eta_{bp} \eta^{aq} \right) \eta^{fg} \varepsilon_{gcde} + \dots \\
		&= \left( \ostar \w{\omega} \right)^k_{\phantom{k} lmnr} \left( \Lambda^{-1} \right)^a_{\phantom{a} k} \Lambda^l_{\phantom{l} b} \Lambda^m_{\phantom{m} c} \Lambda^n_{\phantom{n} d} \Lambda^r_{\phantom{r} e} + \\
		&+ a \left[ \tilde{h}_b^{\phantom{c} \mu} \left( \Lambda^{-1} \right)^a_{\phantom{a} k} \left( \partial_{\mu} \Lambda^k_{\phantom{k} f} \right) -  \eta_{bp} \eta^{aq} \tilde{h}_q^{\phantom{c} \mu} \left( \Lambda^{-1} \right)^p_{\phantom{a} k} \left( \partial_{\mu} \Lambda^k_{\phantom{k} f} \right) \right]  \eta^{fg} \varepsilon_{gcde} + \dots \\
		&= \left( \ostar \w{\omega} \right)^k_{\phantom{k} lmnr} \left( \Lambda^{-1} \right)^a_{\phantom{a} k} \Lambda^l_{\phantom{l} b} \Lambda^m_{\phantom{m} c} \Lambda^n_{\phantom{n} d} \Lambda^r_{\phantom{r} e} + \\
		&+ a \left[ \delta^q_b \delta^a_p -  \eta_{bp} \eta^{aq} \right] \tilde{h}_q^{\phantom{c} \mu} \left( \Lambda^{-1} \right)^p_{\phantom{p} k} \left( \partial_{\mu} \Lambda^k_{\phantom{k} f} \right)  \eta^{fg} \varepsilon_{gcde} + \dots
	\end{split}
\end{align*}
This just means that the dual connection $\left( \ostar \w{\omega}\right)$ has a rather strange trasformation rule. However, there is no problem with that, because the connection 1-form $\w{\omega}$ itself transforms in an unusual manner. It means that the dual connection $\left( \ostar \w{\omega} \right)$ is a different 3-form depending on the basis in which we are working.

This peculiarity is not present in the case of torsion $\w{T}$ or curvature $\Omega$, as all of their indices are essentially \enquote{tensorial} in nature -- $\Omega^a_{\phantom{a} bdc}$ can be thought of as components of a tensor. However, this problem reappears for general $\mathfrak{so}(1,3)$-valued 2-forms or translation algebra-valued 2-forms, in which case we do not know how they transform with respect to the \textit{algebraic} (\enquote{non-tensorial}) indices.

\section[Dual curvature]{Dual curvature $\ostar \Omega$}

Another object we may be interested in may be the dual curvature $\ostar \w{\Omega}$. Note however, that it is pointless to try to examine the behavior of $\ostar$-operator on curvature 2-forms of \textit{teleparallel} connection, because in that case $\w{\Omega} = 0$ and the desired linearity of $\ostar$ immediately gives $\ostar \w{\Omega} = 0$. 

We will therefore consider $\ostar \Omega$, where $\Omega$ is the curvature form of a general (not necessarily teleparallel) connection. Following the same steps as before, we should find all possible contractions of $R^{\alpha}_{\phantom{\alpha} \beta \mu \nu}$ with $\omega_{\alpha \beta \mu \nu}$. Then we parametrize the operator acting on $\Omega$. We get a 2-form $\ostar \Omega$ and we wish to apply the $\ostar$-operator again. But this time, the form under consideration may not have all the symmetries of the Riemann tensor $R^{\alpha}_{\phantom{\alpha} \beta \mu \nu}$. In \cite{Lucas:2008gs}, apparently only the case of curvature of RLC connection is considered, otherwise there would be more (independent) terms in the parametrization of $\ostar \Omega$. We will be more general here and require only that $\Omega$ be the curvature 2-form of a metric-compatible connection $\omega$, and we will point to the flaws in the treatment in \cite{Lucas:2008gs}.

\subsection{Possible contractions and parametrization}

For the curvature tensor of a metric-compatible connection, the following holds:
\begin{align}
	R_{\alpha \beta \mu \nu} &= - R_{\alpha \beta \nu \mu}, \\
	R_{\alpha \beta \mu \nu} &= - R_{\beta \alpha \mu \nu},
\end{align}
hence the possible independent contractions allow us to parametrize 
\begin{equation} \label{dual_curvature_parametrization}
	\begin{aligned}
		\left( \ostar \Omega \right)^{ab}_{\phantom{ab} cd} = \omega_{cdef} &\left[ \alpha R^{abef}  + \beta \left( R^{aebf} - R^{beaf} \right) + \gamma \left( g^{ea} R^{lfb}_{\phantom{lfb} l} - g^{eb} R^{lfa}_{\phantom{lfa} l} \right) + \right. \\
		 &+ \left. \mu g^{ae} g^{bf} R^{lk}_{\phantom{lk} lk} + \lambda R^{efab} + \kappa \left( g^{ea} R^{lbf}_{\phantom{lbf} l} - g^{eb} R^{laf}_{\phantom{laf} l} \right) \right].
	\end{aligned} 
\end{equation}
The terms in brackets are to ensure the pseudo-antisymmetry of $\left( \ostar \Omega \right)^a_{\phantom{a} bcd}$ in indices $a,b$. 
In the case of RLC connection,
\begin{equation} \label{Riemann_index_pair_exchange_symmetry}
	R_{\alpha \beta \mu \nu} = R_{\mu \nu \alpha \beta},
\end{equation}
which makes the term with $\lambda$ the same as the one with $\alpha$, and the term with $\kappa$ the same as the one with $\gamma$. This is the reason that these two terms do not appear in the treatment by Lucas and Pereira \cite{Lucas:2008gs, Lucas:thesis}, where they only have
\begin{equation} \label{dual_curvature_parametrization_Pereira}
	\begin{aligned}
		\left( \ostar \Omega \right)^{ab}_{\phantom{ab} cd} = \omega_{cdef} &\left[ \alpha R^{abef}  + \beta \left( R^{aebf} - R^{beaf} \right) + \right. \\
		&+ \left. \kappa \left( g^{ea} R^{lbf}_{\phantom{lbf} l} - g^{eb} R^{laf}_{\phantom{laf} l} \right) + \mu g^{ae} g^{bf} R^{lk}_{\phantom{lk} lk} \right].
	\end{aligned} 
\end{equation}

Another problem with the treatment by Lucas and Pereira is that $\left( \ostar \Omega \right)$ does not have all the symmetries of RLC curvature tensor -- even if $\Omega$ is the curvature 2-form of RLC connection. When applying $\ostar$ on $\left( \ostar \Omega \right)$, we would have to use a different parametrization from (\ref{dual_curvature_parametrization_Pereira}) to include \textit{all} independent contractions -- and that is exactly what we have done in (\ref{dual_curvature_parametrization}). Let us show that $\left( \ostar \Omega \right)$ lacks the symmetry (\ref{Riemann_index_pair_exchange_symmetry}). We are asking whether
\begin{equation*}
	\left( \ostar \Omega \right)_{abcd} = \omega_{cdef} \left[ \alpha R^{\phantom{ab} ef}_{ab}  + \beta \left( R^{\phantom{a} e \phantom{b} f}_{a \phantom{e} b} - R^{\phantom{b} e \phantom{a} f}_{b \phantom{e} a} \right) + \kappa \left( \delta^e_a R^{l \phantom{b} f}_{\phantom{l} b \phantom{f} l} - \delta^e_b R^{l \phantom{a} f}_{\phantom{l} a \phantom{f} l} \right) + \mu \delta^e_a \delta^f_b R^{lk}_{\phantom{lk} lk} \right]
\end{equation*}
is the same as
\begin{equation*} 
	\left( \ostar \Omega \right)_{cdab} = \omega_{abef} \left[ \alpha R^{\phantom{cd} ef}_{cd} + \beta \left( R^{\phantom{c} e \phantom{d} f}_{c \phantom{e} d} - R^{\phantom{d} e \phantom{c} f}_{d \phantom{e} c} \right) + \kappa \left( \delta^e_c R^{l \phantom{d} f}_{\phantom{l} d \phantom{f} l} - \delta^e_d R^{l \phantom{c} f}_{\phantom{l} c \phantom{f} l} \right) + \mu \delta^e_c \delta^f_d R^{lk}_{\phantom{lk} lk} \right].
\end{equation*}
Obviously, the last terms are the same, so now taking $a=0,b=1,c=2,d=3$, we get
\begin{equation*} 
	\begin{aligned}
		\left( \ostar \Omega \right)_{0123} - \left( \ostar \Omega \right)_{2301} =  \omega_{23ef} \left[ \alpha R^{\phantom{01} ef}_{01}  + \beta \left( R^{\phantom{0} e \phantom{1} f}_{0 \phantom{e} 1} - R^{\phantom{1} e \phantom{0} f}_{1 \phantom{e} 0} \right) + \right. & \left. \kappa \left( \delta^e_0 R^{l \phantom{1} f}_{\phantom{l} 1 \phantom{f} l} - \delta^e_1 R^{l \phantom{0} f}_{\phantom{l} 0 \phantom{f} l} \right) \right] \\
		- \omega_{01ef} \left[ \alpha R^{\phantom{23} ef}_{23} + \beta \left( R^{\phantom{2} e \phantom{3} f}_{2 \phantom{e} 3} - R^{\phantom{3} e \phantom{2} f}_{3 \phantom{e} 2} \right) + \right. & \left. \kappa \left( \delta^e_2 R^{l \phantom{3} f}_{\phantom{l} 3 \phantom{f} l} - \delta^e_3 R^{l \phantom{2} f}_{\phantom{l} 2 \phantom{f} l} \right) \right],
	\end{aligned} 
\end{equation*}
\begin{equation*} 
	\begin{aligned}
		\left( \ostar \Omega \right)_{0123} &- \left( \ostar \Omega \right)_{2301} \\
		&= \left[ 2 \alpha R^{\phantom{01} 01}_{01}  + \beta \left( R^{\phantom{0} 0 \phantom{1} 1}_{0 \phantom{0} 1} - R^{\phantom{0} 1 \phantom{1} 0}_{0 \phantom{1} 1} - R^{\phantom{1} 0 \phantom{0} 1}_{1 \phantom{0} 0} + R^{\phantom{1} 1 \phantom{0} 0}_{1 \phantom{1} 0} \right) + \kappa \left( R^{l \phantom{1} 1}_{\phantom{l} 1 \phantom{1} l} + R^{l \phantom{0} 0}_{\phantom{l} 0 \phantom{0} l} \right) \right] \\
		&- \left[ 2 \alpha R^{\phantom{23} 23}_{23} + \beta \left( R^{\phantom{2} 2 \phantom{3} 3}_{2 \phantom{2} 3} - R^{\phantom{2} 3 \phantom{3} 2}_{2 \phantom{3} 3} - R^{\phantom{3} 2 \phantom{2} 3}_{3 \phantom{2} 2} + R^{\phantom{3} 3 \phantom{2} 2}_{3 \phantom{3} 2} \right) + \kappa \left( R^{l \phantom{3} 3}_{\phantom{l} 3 \phantom{3} l} + R^{l \phantom{2} 2}_{\phantom{l} 2 \phantom{2} l} \right) \right],
	\end{aligned} 
\end{equation*}
which generally does not vanish. The $\alpha$-term and $\beta$-term from the first line is nowhere to be found in the second line, so to make this identically zero, we would need $\alpha = \beta = 0$. Our question reduces to whether
\begin{equation} 
	\begin{aligned}
		R_{01}^{\phantom{01} 10} + R_{21}^{\phantom{21} 12} &+ R_{31}^{\phantom{31} 13} + R_{10}^{\phantom{10} 01} + R_{20}^{\phantom{20} 02} + R_{30}^{\phantom{30} 03} \\
		&- \left( R_{03}^{\phantom{03} 30} + R_{13}^{\phantom{13} 31} + R_{23}^{\phantom{23} 32} + R_{02}^{\phantom{02} 20} + R_{12}^{\phantom{12} 21} + R_{32}^{\phantom{32} 23} \right)
	\end{aligned} 
\end{equation}
vanishes, which is equivalent to asking whether
\begin{equation}
		R_{10}^{\phantom{10} 01} = R_{32}^{\phantom{32} 23},
\end{equation}
which, obviously, does not hold in general (even in the case of RLC connection).

Note that in the entire reasoning above, we did not need to use (\ref{Riemann_index_pair_exchange_symmetry}), but even if we had attempted to use it, it would have been of no help to us in showing that $\left( \ostar \Omega \right)_{0123} - \left( \ostar \Omega \right)_{2301}$ vanishes.

\subsection{Duality condition}

Applying $\ostar$ to $\Omega$ twice, putting this into \textit{Mathematica} and requiring $\left( \ostar \ostar \Omega \right) = - \Omega$ yields a set of four non-linear equations
\begin{align}
	0 &= \beta \gamma + \alpha \mu - 2\beta \mu + \alpha \lambda + \beta \lambda - \gamma \lambda + \mu \lambda + \lambda^2 + \beta \kappa - \lambda \kappa, \\
	0 &= \alpha \beta + \beta^2 + \alpha \gamma - 2 \beta \gamma - 2\alpha \mu + 4 \beta \mu + \beta \lambda + \gamma \lambda - 2\mu \lambda + \alpha \kappa - 2\beta \kappa + \lambda \kappa, \\
	0 &= 2\alpha \beta + \beta^2 - 2\alpha \gamma - \gamma^2 + 4\alpha \lambda + 2\beta \lambda + 2\gamma \lambda + 2\alpha \kappa + 2\gamma \kappa - 2\lambda \kappa - \kappa^2, \\
	\dfrac{1}{4} &= \alpha^2 + \alpha \beta - \alpha \gamma + \beta \gamma + \alpha \mu - 2\beta \mu + \alpha \lambda + \mu \lambda - \alpha \kappa + \beta \kappa. 
\end{align}

Unsurprisingly, these equations do not have a unique solution. And, also unsurprisingly, because it is just a special case of ordinary Hodge operator, one of these solutions is $\alpha = \frac{1}{2}, \beta=\gamma=\mu=\lambda=\kappa=0$, which Lucas and Pereira present in \cite{Lucas:2008gs}.

\subsection{Applying the additional rule}

If we apply the additional rule described in section \ref{add_rule}, we have that $\lambda = \alpha, \beta = \pm 2 \alpha$ and $\kappa = \pm \gamma$. As the $\pm$ signs in $\beta$ and $\kappa$ are independent, let us write
\begin{align}
	\beta = (-1)^{\xi} 2 \alpha,   &&   \kappa = (-1)^{\zeta} \gamma,
\end{align}
where $\xi, \zeta \in \{0,1\}$. After substituting for $\lambda$ and $\kappa$ the equations reduce to
\begin{align}
	0 &= 2\alpha \mu - 2\beta \mu + 2\alpha^2 + \alpha \beta + \left[ 1 + (-1)^{\zeta} \right] \beta \gamma - \left[ 1 + (-1)^{\zeta} \right] \alpha \gamma, \\
	0 &= 2\alpha \beta + \beta^2 - 4\alpha \mu + 4 \beta \mu - 2 \left[ 1 + (-1)^{\zeta} \right] \beta \gamma + 2 \left[ 1 + (-1)^{\zeta} \right] \alpha \gamma, \\
	0 &= 4\alpha \beta + \beta^2 + 4\alpha^2 + 2 \left[ -1 + (-1)^{\zeta} \right] \gamma^2, \\
	\dfrac{1}{4} &= 2\alpha^2 + \alpha \beta + 2\alpha \mu - 2\beta \mu - \left[ 1 + (-1)^{\zeta} \right] \alpha \gamma + \left[ 1 + (-1)^{\zeta} \right] \beta \gamma. 
\end{align}
Taking twice the first equation and adding it to the second one gives
\begin{equation}
	0 = 4 \alpha^2 + 4\alpha \beta + \beta^2 = \left( 2\alpha + \beta \right)^2, 
\end{equation}
leading to $\beta = -2 \alpha$, or equivalently, $\xi = 1$. However, looking at the system more closely, it is clear that the first and the last equation are contradictory.

\subsection{Dual RLC curvature}

Now, not only does (\ref{dual_curvature_parametrization_Pereira}) not contain all possible contractions, but also it is not clear why in \cite{Lucas:2008gs}, the term with $\kappa$ is preferred over the term with $\gamma$, or the term with $\alpha$ over the term with $\lambda$. Moreover, using parametrization (\ref{dual_curvature_parametrization_Pereira}) while assuming that $\Omega$ is the curvature of a RLC connection and doing the same procedure as we have done starting from (\ref{dual_curvature_parametrization}), we arrive at equations which are equivalent to the system
\begin{align}
	\left( \alpha + \beta \right)^2 &= \dfrac{1}{4}, \\
	\beta^2 + \alpha \beta + \alpha \kappa &= 0, \\
	\alpha \mu + \beta \kappa - 2 \beta \mu &= 0,
\end{align}
which has infinitely many solutions:
\begin{align}
	\alpha &= \dfrac{1}{2} \left( \pm 1 -2 \beta \right), \\
	\mu &= \dfrac{2 \beta^2}{\pm 1 - 8\beta \pm 12\beta^2}, \\
	\kappa &= \frac{\beta}{-1 \pm 2\beta}.
\end{align}
Interestingly, repeating the process for parametrization
\begin{equation}
	\begin{aligned}
		\left( \ostar \Omega \right)^{ab}_{\phantom{ab} cd} = \omega_{cdef} &\left[ \alpha R^{abef}  + \beta \left( R^{aebf} - R^{beaf} \right) + \right. \\
		&+ \gamma \left( g^{ea} R^{lfb}_{\phantom{lfb} l} - g^{eb} R^{lfa}_{\phantom{lfa} l} \right) + \left. \mu g^{ae} g^{bf} R^{lk}_{\phantom{lk} lk} \right]
	\end{aligned} 
\end{equation}
leads to the same sytem of equations -- in these equations, $\kappa$ is replaced by $\gamma$. 


To summarize, there are numerous unclear steps or unstated assumptions in the article by Lucas and Pereira. Attempting to clarify and do the calculation properly -- according to the rules stated in section \ref{contraction_rules} -- leads to a system of nonlinear equations, and the construction of the duality operator is not unique. Our calculation for the curvature of RLC connection also led to a non-unique solution. 

	\chapter{Applications} \label{chapter6}

\section{Teleparallel action as Yang-Mills action for gauge theory}

The Yang-Mills action for a gauge theory whose field strength is $\mathcal{F}$ takes the form 
\begin{align}
	\mathcal{S}_{YM} = - \int_{\mathcal{M}}{\Tr \left( \mathcal{F} \wedge \star \mathcal{F} \right)}.
\end{align}
We have already mentioned in section \ref{action__dual_rewrite} that the $\ostar$-operator may be used to rewrite the action of teleparallel gravity in the form
\begin{align}
	\w{\mathcal{S}} &= \int_{\mathcal{M}}{\frac{1}{2 \kappa} \eta_{ab} \left( \w{T}^a \wedge \ostar \w{T}^b \right)},
\end{align}
which looks almost like the Yang-Mills action for a gauge theory, with torsion $\w{T}$ being the field strength of a hypothetical gauge field. Here, the Hodge $\star$ operator is replaced by the $\ostar$-operator.

It is also possible to formulate teleparallel gravity as a gauge theory of translations \cite{Pereira:2019woq,Cho:1976, Bahamonde:2021gfp, Pereira}.

\subsection{Teleparallel gravity as gauge theory -- motivation}

Looking at Cartan structure equations (\ref{Cartan_struct_torsion}), (\ref{Cartan_struct_curvature}), we see that they are both similar. The right-hand side contains two terms, one of which is the exterior derivative $d$ of some differential form, while the other is the exterior product of certain forms. This might inspire us to explore the question whether this is just a coincidence, or there are some deeper reasons for this similarity. This question is even more urgent if we recall that in the theory of connections on principal bundles \cite{Fecko}, for a principal bundle $\pi : P \longrightarrow \mathcal{M}$ the curvature 2-form $F \in \Omega^2 \left(P, \Ad \right)$ for the connection $A \in \Omega^1 \left(P, \Ad \right) $ is given by
\begin{equation}
	\begin{aligned}
		F &= D A = d A + \dfrac{1}{2} \left[ A \wedge A \right] = d A^i E_i + \dfrac{1}{2} A^j \wedge A^k \left[ E_j , E_k \right],
	\end{aligned}
\end{equation}
which also bears some similarity to (\ref{Cartan_struct_torsion}).

Is it possible to construct a principal bundle in such a way that the curvature $F$ of the gauge connection $A$ on this bundle would be given by (\ref{Cartan_struct_torsion}), i.e. the curvature of such a connection would actually be torsion of an affine connection? 

Within the framework of gauge theories, the connection is related to the gauge field $\mathcal{A}$ via a representation $\rho^{\prime}$ of the Lie algebra $\mathcal{G}$ of the structure group $G$ of $P$, i.e. $\mathcal{A} = \sigma^*A^i \rho^{\prime} \left( E_i \right)$, where $\sigma$ is a local section. Then the field strength is
\begin{equation} \label{field_strength}
	\mathcal{F} = d \mathcal{A}^i \rho^{\prime} \left( E_i \right) + \dfrac{1}{2} \mathcal{A}^i \wedge \mathcal{A}^j \left[ \rho^{\prime} \left( E_i \right), \rho^{\prime} \left( E_j \right) \right] 
\end{equation}
Comparing this to 
\begin{equation} \label{Cartan_torsion}
	T^a = de^a + \omega^a_{\phantom{a} b} \wedge e^b,
\end{equation}
where $e^a$ is a coframe (not necessarily a tetrad) and $\omega^a_{\phantom{a} b}$ are connection 1-forms with respect to frame $e_a$, we hypothesize that $\mathcal{A}^i$ is somehow related to $h^a$. However, in the gauge theoretic view, there is no place for the affine connection $\omega^a_{\phantom{a} b}$. It might seem like the game's over. But we don't give up! To make (\ref{field_strength}) the same as (\ref{Cartan_torsion}), we have to get rid of the terms with $\wedge$. This can be done in (\ref{field_strength}) by considering an abelian $\mathcal{G}$, and in (\ref{Cartan_torsion}) by setting $\omega^a_{\phantom{a} b} \equiv 0$. 

What does this correspond to? First, notice that $a$ is an index taking on $4$ possible values: $0,1,2,3$. Thus, the structure group $G$ must be an abelian group of dimension $4$, so that index $i$ also runs through $4$ possible values. Second, we have not said that there is any metric defined on the manifold $\mathcal{M}$. We may therefore define $e_a$ to be a tetrad and construct a metric using (\ref{metric_from_tetrad}). Then the condition $\omega^a_{\phantom{a} b} \equiv 0$ would mean that the connection is actually a metric teleparallel connection, $e_a$ being just the tetrad with respect to which the connection 1-forms are vanishing. 

\subsection[Gauge theory of translation group]{Gauge theory of translation group $\mathbb{T}_4$}

So far, we have worked backwards. To reverse our reasoning, we have to find a $4$-dimensional abelian Lie group. The translation group $\mathbb{T}_4$ is a good candidate.


In fact, any connected abelian real Lie group $G$ is isomorphic to $\mathbb{R}^k \times \left( S^1 \right)^h$ for some $h,k \in \mathbb{N} \cup \lbrace 0 \rbrace$ \cite{Knapp1988LieGB}. If the $G$ is, in addition, compact, then it is isomorphic to $\left( S^1 \right)^h$. A motivation for choosing the translational group $\mathbb{T}_4$ is that by Noether's theorem, energy-momentum conservation is induced by translational invariance of the Lagrangian of an isolated system \cite{Hehl:2020hhp, Feynman}.


On a principal $\mathbb{T}_4$-bundle $P$ there exists a connection $A \in \Omega^1  \left( P, \Ad \right)$, whose curvature is
\begin{equation}
	F = D A = d A + \dfrac{1}{2} \left[ A \wedge A \right] = dA.
\end{equation}
Consider a local section $\sigma : \mathcal{M} \supset \mathcal{O} \longrightarrow P$. We define
\begin{equation} \label{gauge_tetrad}
	h^a \coloneqq \sigma^* A^a, 
\end{equation}
and require that this be the cotetrad such that $h_a$ is the tetrad with respect to which the (metric) teleparallel connection vanishes: $\w{\omega}^a_{\phantom{a} b} = 0$. Then the metric on $\mathcal{M}$ is 
\begin{equation}
	g = \eta_{ab} h^a \otimes h^b .
\end{equation}
Also, by requiring that connection 1-forms $\w{\omega}^a_{\phantom{a} b}$ transform in the way in which they should when changing the tetrad to $\tilde{h}_a (x) = \Lambda^b_{\phantom{b} a} (x) h_b$, via a Lorentz transformation $\Lambda (x) \in SO (1,3)$, we get 
\begin{equation}
	\tilde{\w{\omega}}^a_{\phantom{a} b} = \left( \Lambda^{-1}\right)^a_{\phantom{a} c} d\Lambda^c_{\phantom{c} b}.
\end{equation}
This is how the teleparallel connection enters the game when describing teleparallel gravity as a gauge theory of the translation group.

The field strength is defined as
\begin{equation}
	\mathcal{F}^a = \sigma^* F^a = d \left( \sigma^* A^a \right) = dh^a
\end{equation}
and has components
\begin{equation}
	\mathcal{F}^a_{\mu \nu} = dh^a \left( \partial_{\mu}, \partial_{\nu} \right) = \partial_{\mu} h^a \left( \partial_{\nu} \right) - \partial_{\nu} h^a \left( \partial_{\mu} \right) - h^a \left( \left[ \partial_{\mu}, \partial_{\nu} \right] \right)  = \partial_{\mu} h^a_{\phantom{a} \nu} - \partial_{\nu} h^a_{\phantom{a} \mu},
\end{equation}
while the components of torsion are 
	\begin{align*}
	\w{T} \left( h^a , \partial_{\mu}, \partial_{\nu} \right) &=  \langle h^a, \w{\nabla}_{\partial_{\mu}} \partial_{\nu} - \w{\nabla}_{\partial_{\nu}} \partial_{\mu} - \left[ \partial_{\mu}, \partial_{\nu} \right] \rangle \\
	&= \langle h^a, \w{\nabla}_{\partial_{\mu}} \left( h^b_{\phantom{b} \nu} h_b \right) - \w{\nabla}_{\partial_{\nu}} \left( h^b_{\phantom{b} \mu} h_b \right) \rangle \\
	&= \langle h^a, \left( \partial_{\mu} h^b_{\phantom{b} \nu} \right) h_b + h^b_{\phantom{b} \nu} \w{\omega}^c_{\phantom{c} b \mu} h_c - \left( \partial_{\nu} h^b_{\phantom{b} \mu} \right) h_b - h^b_{\phantom{b} \mu} \w{\omega}^c_{\phantom{c} b \nu} h_c \rangle \\
	&= \partial_{\mu} h^a_{\phantom{a} \nu} - \partial_{\nu} h^a_{\phantom{a} \mu}.
\end{align*}
because $\w{\omega}^a_{\phantom{a} b \mu} \equiv 0$. Therefore
\begin{equation}
	\mathcal{F}^a_{\mu \nu} = \w{T}^a_{\mu \nu}.
\end{equation}
Using a different tetrad $\tilde{h}_a = \Lambda^b_{\phantom{b} a} (x) h_b$ leads to
\begin{align*}
	\mathcal{F}^a &= dh^a = d \left( \Lambda^a_{\phantom{a} b} \tilde{h}^b \right) = \Lambda^a_{\phantom{a} b} d\tilde{h}^b + d \Lambda^a_{\phantom{a} b} \wedge \tilde{h}^b \\
	&= \Lambda^a_{\phantom{a} b} \left( d\tilde{h}^b + \left( \Lambda^{-1} \right)^b_{\phantom{b} c} d \Lambda^c_{\phantom{c} d} \wedge \tilde{h}^d \right) = \Lambda^a_{\phantom{a} b} \left( d\tilde{h}^b + \tilde{\w{\omega}}^b_{\phantom{b} d} \wedge \tilde{h}^d \right) = \Lambda^a_{\phantom{a} b} \tilde{\w{T}}^b ,
\end{align*}
where $\tilde{\w{T}}^a$ are torsion 2-forms with respect to $\tilde{h}_a$. All of this should be also clear from the Cartan structure equation (\ref{Cartan_struct_torsion}).

To sum up, any choice of frame generates a metric if we declare that this frame is a tetrad, and it also generates a (metric) teleparallel connection.

There is a subtle problem with the consistency of this clever construction, however. The tetrad with respect to which the connection vanishes is given by eq. (\ref{gauge_tetrad}). Using a different local section $\sigma^{\prime}$ leads to a different frame $h_a^{\prime}$ and may result in different metric $g^{\prime}$ and different teleparallel connection $\w{\omega}^{\prime a}_{\phantom{\prime a} b}$.

The gauge structure of teleparallel gravity is a topic that is still being discussed, see for example \cite{Pereira:2019woq,LeDelliou:2019esi}

\section[Teleparallel gravity as analogue of electrodynamics]{Teleparallel gravity as analogue of electrodynamics} \label{el_dyn}

In the language of differential forms, the general electrodynamics can be formulated in terms of the electromagnetic field strength 2-form $F$, the closed current 3-form $J$ and the excitation 2-form $H$. 
General electrodynamics is then described by the field equations:
\begin{align} \label{Maxwell_eqs}
	dF = 0,  &&  d H = J,
\end{align}
as well as a relation between field strength and excitation called \textit{constitutive relation} $H = H(F)$. This approach allows for the description of the vast phenomenology of the electromagnetic interaction, from linear and non-linear media to various non-linear theories of electrodynamics, like Born-Infeld theory, in one unified language \cite{Itin:premetric}. 

In the simplest case of Maxwell electrodynamics in vacuum, the excitation form is related to the field strength via the Hodge dual operator and the vacuum impedance $\lambda_0 = \sqrt{\dfrac{\varepsilon_0}{\mu_0}}$
\begin{equation} \label{EM_excitation}
	H = \lambda_0 \star F.
\end{equation}
In the so-called premetric approach described above, the constitutive relation can be of any kind and does not necessarily involve a metric tensor.

In teleparallel gravity, torsion is defined by the first Cartan structure equation 
\begin{equation}
	T^a \equiv D h^a = dh^a + \omega^a_{\phantom{a} b} \wedge h^b,
\end{equation}
where $D = d + \omega^a_{\phantom{a} b} \wedge$ is exterior covariant derivative\footnote{This equation can be better understood in the following way: on the bundle of orthonormal frames $O\mathcal{M}$, there exists the so-called canonical $\mathbb{R}^n$-valued 1-form $\theta$ \cite{Fecko, Kobayashi} (which is the solder form), and a section $\sigma$ corresponds to a frame $h_a$. Then $\sigma^* \theta^a = h^a$. Understood in this way, it makes sense to apply the exterior covariant derivative $D$ on $h^a$.
	}. 
	Here, $T^a$ are 2-forms and $\omega^a_{\phantom{a}b}$ are 1-forms which satisfy the condition of vanishing curvature expressed by the second Cartan structure equation
\begin{equation} \label{zero_curvature}
	d \omega^a_{\phantom{a}b} + \omega^a_{\phantom{a}c} \wedge \omega^c_{\phantom{c}b} \equiv 0.
\end{equation}
Taking the exterior covariant derivative of the torsion and using (\ref{zero_curvature}) we obtain the Bianchi identity
\begin{equation} \label{Bianchi_torsion}
	D T^a = 0,
\end{equation}
which is an analogue of the first Maxwell equation in (\ref{Maxwell_eqs}). 

Setting $\kappa = 8 \pi G = 1$, the field equations of teleparallel gravity are 
\begin{equation} \label{field_eq_TEGR}
	D \tilde{H}_a - \tilde{\Upsilon}_a = \Sigma_a 
\end{equation}
and they play the role of the second equation in (\ref{Maxwell_eqs}). The teleparallel gravity excitation 2-forms $H_a$ are given by
\begin{equation}
	\tilde{H}_a \left( h^a, T^a \right) = \dfrac{1}{4} h \varepsilon_{\alpha \beta \rho \sigma} S_a^{\phantom{a} \rho \sigma} dx^{\alpha} \wedge dx^{\beta},
\end{equation}
where $S_a^{\phantom{a} \rho \sigma}$ is the superpotential (\ref{superpotential}), $h = \lvert \det (h^a_{\phantom{a} \mu}) \rvert$, and  $\Upsilon_a$ is the gravitational energy-momentum 3-form and $\Sigma_a$ the matter energy-momentum 3-form.
The field equations (\ref{field_eq_TEGR}) are obtained by varying the action $\tilde{S} \left[ h^a, \omega^a_{\phantom{a} b}, \Psi_I \right] = \tilde{S}_g \left[ h^a, \omega^a_{\phantom{a} b} \right] + S_m \left[ h^a, \Psi_I \right]$, where the gravitational action is
\begin{equation}
	\tilde{S}_g = \dfrac{1}{2} \int_{\mathcal{M}}{T^a \wedge \tilde{H}_a} = \dfrac{1}{4} \int_{\mathcal{M}}{T^a_{\phantom{a} \mu \nu} S_a^{\phantom{a} \mu \nu} h d^4 x}
\end{equation}
and $S_m$ describes the coupling between gravity and matter fields $\Psi_I$.

The use of the language of premetric electrodynamics can be extended from local and linear theories to local, but generally non-linear teleparallel gravity theories by consindering gravitational excitation tensors $H_a$ which are general functions of the tetrad and the torsion: $H_a = H_a \left( h^b, T^b \right)$. By varying the action $S = S_g + S_m$ with the generalized gravitational action
\begin{equation}
	S_g \left[ h^a, \omega^a_{\phantom{a} b} \right] = \dfrac{1}{2} \int_{\mathcal{M}}{T^a \wedge H_a}
\end{equation}
the field equations are obtained
\begin{equation}
	D \Pi_a - \Upsilon_a = \Sigma_a, 
\end{equation}
where $\Pi_a$ are 2-forms obtained from $H_a$, and $\Upsilon_a$ is analogous to $\tilde{\Upsilon}_a$. The Bianchi identity (\ref{Bianchi_torsion}) is the same as in the previous case \cite{Hohmann:2017duq}. 

In the case of teleparallel equivalent of general relativity discussed above, $\Pi_a = H_a = \tilde{H}_a$. 
Notice that 
\begin{equation}
	\tilde{H}_a = \eta_{ab} \left( \ostar \w{T} \right)^b. 
\end{equation}
This constitutive relation in teleparallel gravity theories is very similar to (\ref{EM_excitation}) in electrodynamics. If the $\ostar$ could be defined properly as an operator, the teleparallel equivalent of general relativity would be an analogue of Maxwell electrodynamics in vacuum. Thus, the teleparallel equivalent of general relativity would gain the same importance among the various teleparallel gravity theories that vacuum Maxwell electrodynamics has among the various theories of premetric electrodynamics.

	\backmatter
	
	\chapter{Conclusions} 
\pagestyle{plain}

The description of gravity in terms of torsion of a metric teleparallel connection and using the tetrad formalism proved to be a very clever idea. From the mathematical perspective, TEGR is an elegant theory of gravity, which can be formulated as a gauge theory of the group $\mathbb{T}_4$ of translations in 4D, as we have seen in chapter \ref{chapter6}.

The main focus of this thesis was the search for a proper definition of a duality operator $\ostar$, whose existence was proposed by Lucas and Pereira \cite{Lucas:2008gs}. 
We have inquired into the problem whether this $\ostar$ operator can be constructed consistently from the mathematical point of view. Our analysis of the results of Lucas and Pereira \cite{Lucas:2008gs}, together with our motivation for $\ostar$ to be similar to the Hodge $\star$ operator, led to our formulation of an algorithm for its construction in a general case, i.e. for the application of $\ostar$ on any $p$-form $\alpha$ on $n$-dimensional manifold $\mathcal{M}$. The cases of $p$-form $\alpha$ with $q$ algebraic indices (see section \ref{contraction_rules}) and $(n-p)$-form $\beta$ with $q$ algebraic indices have to be considered simultaneously. First, all the possible contractions of $\alpha$ and $\beta$ with the volume form on $\mathcal{M}$ are determined, according to rules stated in section \ref{contraction_rules}. After that, we parametrize $\ostar \alpha$ using these contractions and real parameters. Next, we impose duality conditions
\begin{align*}
	\ostar \ostar \alpha &= \sgn g (-1)^{p(n+1)} \alpha, \\
	\ostar \ostar \beta &= \sgn g (-1)^{(n-p)(n+1)} \beta,
\end{align*}
which yield a system of algebraic equations for the parameters.

We followed this algorithm for three cases: the torsion 2-forms $T^a$, the connection 1-forms $\omega^a_{\phantom{a} b}$ and for curvature 2-forms $\Omega^a_{\phantom{a} b}$ on a $4$-dimensional manifold with a metric tensor of signature $(1,3)$. The systems of equations did not have a unique solution in any of these cases (not even up to a sign).

To remedy this non-uniqueness, we devised another rule, see section \ref{add_rule}, which for torsion produced the result obtained in \cite{Lucas:2008gs}, but for connection 1-forms and curvature 2-forms failed to provide any solution at all. This is contrary to the claims of Lucas and Pereira, who assert they have obtained a unique solution for the case of curvature (up to a sign).

Thus, to sum up our results, we showed that the naive approach taken by Lucas and Pereira does not lead to a consistent definition of this operator, because it results in the operator with the desired properties either being non-unique or non-existent.

The Hodge $\star$ operator \enquote{ignores} the algebraic indices (see section \ref{contraction_rules}) of the forms. Thus, any attempt at a definition of the newly-proposed duality operator $\ostar$ that is based on utilizing the Hodge $\star$ operator at some point, no matter how sophisticated and clever the precise prescription, formula or procedure is, is destined for failure. 

For example, one approach to the construction might leverage the geometry of the frame bundle endowed with a connection $\omega$ and use the canonical $\mathbb{R}^n$-valued 1-form $\theta$ (which is the solder form) \cite{Fecko} and connection 1-forms to construct a (canonical) volume form and a (canonical) metric tensor, and use these in the definition of the Hodge $\star$ operator for forms on the frame bundle. This $\star$ operator could then be applied on forms on the frame bundle (because connection $\omega$, torsion $T = D \theta$ and curvature $\Omega = D \omega$ are forms defined on the frame bundle -- see section \ref{el_dyn}), and we could apply pull-back by a local section $\sigma$ on the result, so that we could obtain a result on the base manifold of the bundle. However, the $\star$ operator constructed and used in this way also \enquote{ignores} the algebraic indices.

For that reason, a deeper insight into the geometry underlying teleparallel gravity, frame bundles, soldered bundles or principal fiber bundles is necessary (but might not be sufficient) to obtain a proper coordinate-free definition of the operator, from which the geometric interpretation would be clear and such a definition would be canonical.

	\begin{appendices}
	\pagestyle{plain}
	\chapter[Calculation of dual torsion]{Calculation of $\left( \ostar \ostar \w{T} \right)$} \label{append:torsion}
	In section \ref{dual_torsion} we parametrized $\left( \ostar \w{T} \right)$ using three real parameters $a,b,c$. We now have to calculate $\left( \ostar \ostar \w{T} \right) $\footnote{The result of this calculation is independent of the choice of frame, therefore we can work in a holonomic frame $\partial_{\alpha}$. }:
	\begin{equation}
		\begin{aligned}
			\left( \ostar \ostar \w{T} \right)^{\gamma}_{\phantom{\gamma} \alpha \beta} &=  \omega_{\alpha \beta \mu \nu} \left[ a \left( \ostar \w{T} \right)^{\gamma \mu \nu} + b \left( \ostar \w{T} \right)^{\mu \gamma \nu} + c \left( \ostar \w{T} \right)^{\lambda \mu}_{\phantom{\lambda \mu} \lambda} g^{\nu \gamma} \right] \\
			&= \omega_{\alpha \beta \mu \nu} \left[ a \omega^{\mu \nu \rho \sigma} \left( a \w{T}^{\gamma}_{\phantom{\gamma} \rho \sigma} + b \w{T}^{\phantom{\rho} \gamma}_{\rho \phantom{\gamma} \sigma} + c \w{T}^{\lambda}_{\phantom{\lambda} \rho \lambda} \delta_{\sigma}^{\gamma} \right) \right. \\
			&+ \left. b \omega^{\gamma \nu \rho \sigma} \left( a \w{T}^{\mu}_{\phantom{\mu} \rho \sigma} + b \w{T}^{\phantom{\rho} \mu}_{\rho \phantom{\mu} \sigma} + c \w{T}^{\lambda}_{\phantom{\lambda} \rho \lambda} \delta_{\sigma}^{\mu} \right)  \right. \\
			& + \left. c \omega^{\mu \phantom{\lambda} \rho \sigma}_{\phantom{\mu} \lambda} \left( a \w{T}^{\lambda}_{\phantom{\lambda} \rho \sigma} + b \w{T}^{\phantom{\rho} \lambda}_{\rho \phantom{\lambda} \sigma} + c \w{T}^{\kappa}_{\phantom{\kappa} \rho \kappa} \delta_{\sigma}^{\lambda} \right) g^{\nu \gamma} \right].
		\end{aligned}
	\end{equation}
	The last term vanishes, because $\omega^{\mu \phantom{\lambda} \rho \sigma}_{\phantom{\mu} \lambda} \delta_{\sigma}^{\lambda} = \omega^{\mu \phantom{\sigma} \rho \sigma}_{\phantom{\mu} \sigma} = 0$. Now, remember that $\omega^{\mu \nu \gamma \sigma} = \dfrac{\sgn g}{h} \varepsilon^{\mu \nu \gamma \sigma} = - \dfrac{1}{h} \varepsilon^{\mu \nu \gamma \sigma}$ and $\omega_{\mu \nu \gamma \sigma} = h \varepsilon_{\mu \nu \gamma \sigma}$. Therefore
	\begin{equation}
		\begin{aligned}
			- \left( \ostar \ostar \w{T} \right)^{\gamma}_{\phantom{\gamma} \alpha \beta} &= 4a \delta^{\rho}_{\left[ \alpha \right.} \delta^{\sigma}_{\left. \beta \right] } \left( a \w{T}^{\gamma}_{\phantom{\gamma} \rho \sigma} + b \w{T}^{\phantom{\rho} \gamma}_{\rho \phantom{\gamma} \sigma} + c \w{T}^{\lambda}_{\phantom{\lambda} \rho \lambda} \delta_{\sigma}^{\gamma} \right)  \\
			&+ 3! b \delta^{\gamma}_{\left[ \alpha \right.} \delta^{\rho}_{\left. \beta \right.} \delta^{\sigma}_{\left. \mu \right]} \left( a \w{T}^{\mu}_{\phantom{\mu} \rho \sigma} + b \w{T}^{\phantom{\rho} \mu}_{\rho \phantom{\mu} \sigma} + c \w{T}^{\lambda}_{\phantom{\lambda} \rho \lambda} \delta_{\sigma}^{\mu} \right)  \\
			& + 3! c \delta^{\delta}_{\left[ \alpha \right.} \delta^{\rho}_{\left. \beta \right.} \delta^{\sigma}_{\left. \nu \right]} g_{\delta \lambda} \left( a \w{T}^{\lambda}_{\phantom{\lambda} \rho \sigma} + b \w{T}^{\phantom{\rho} \lambda}_{\rho \phantom{\lambda} \sigma} \right) g^{\nu \gamma} \\
			&= \left( 4 a^2 \w{T}^{\gamma}_{\phantom{\gamma} \alpha \beta} + 4ab \w{T}^{\phantom{\rho} \gamma}_{\left[ \alpha \phantom{\gamma} \beta \right]} - 4ac \w{T}^{\lambda}_{\phantom{\lambda} \lambda \left[ \alpha \right.} \delta_{\left. \beta \right] }^{\gamma} \right) \\
			&+ 6b \left( a \delta^{\gamma}_{\left[ \alpha \right.} \w{T}^{\mu}_{\phantom{\mu} \left. \beta \mu \right]} + b \delta^{\gamma}_{\left[ \alpha \right.} \w{T}^{\phantom{\rho} \mu}_{\left. \beta \phantom{\mu}  \mu \right]} + c \delta_{\left[ \mu \right.}^{\mu} \delta^{\gamma}_{\left. \alpha \right.} \w{T}^{\lambda}_{\phantom{\lambda} \left. \beta \right] \lambda} \right) \\
			& + 6c \left( a \w{T}_{\left[ \alpha \beta \nu \right]} + b \w{T}_{\left[ \beta \alpha \nu \right]} \right) g^{\nu \gamma}.
		\end{aligned}
	\end{equation}
	Notice that $a \w{T}_{\left[ \alpha \beta \nu \right]} + b \w{T}_{\left[ \beta \alpha \nu \right]} = \left(a - b \right) \w{T}_{\left[ \alpha \beta \nu \right]}$. At this point, we have to write all terms and notice that some of them cancel:
	\begin{align*}
		- \left( \ostar \ostar \w{T} \right)^{\gamma}_{\phantom{\gamma} \alpha \beta} &= \left( 4 a^2 \w{T}^{\gamma}_{\phantom{\gamma} \alpha \beta} + 2ab \w{T}^{\phantom{\rho} \gamma}_{\alpha \phantom{\gamma} \beta} - 2ab \w{T}^{\phantom{\rho} \gamma}_{\beta \phantom{\gamma} \alpha} + 2ac \delta_{\alpha}^{\gamma} \w{T}^{\lambda}_{\phantom{\lambda} \lambda \beta} - 2ac \delta_{\beta}^{\gamma} \w{T}^{\lambda}_{\phantom{\lambda} \lambda \alpha} \right) \\
		&+ \left[ ab \left( \delta^{\gamma}_{\alpha} \w{T}^{\mu}_{\phantom{\mu} \beta \mu} - \delta^{\gamma}_{\alpha} \w{T}^{\mu}_{\phantom{\mu} \mu \beta} + \delta^{\gamma}_{\mu} \w{T}^{\mu}_{\phantom{\mu} \alpha \beta} - \delta^{\gamma}_{\mu} \w{T}^{\mu}_{\phantom{\mu} \beta \alpha} + \delta^{\gamma}_{\beta} \w{T}^{\mu}_{\phantom{\mu} \mu \alpha} - \delta^{\gamma}_{\beta} \w{T}^{\mu}_{\phantom{\mu} \alpha \mu} \right) \right. \\
		&+ \left. b^2 \left( \delta^{\gamma}_{\alpha} \underbrace{\w{T}^{\phantom{\rho} \mu}_{\beta \phantom{\mu}  \mu}}_{=0} - \delta^{\gamma}_{\alpha} \w{T}^{\phantom{\rho} \mu}_{\mu \phantom{\mu} \beta} + \delta^{\gamma}_{\mu} \w{T}^{\phantom{\rho} \mu}_{\alpha \phantom{\mu} \beta} - \delta^{\gamma}_{\mu} \w{T}^{\phantom{\rho} \mu}_{\beta \phantom{\mu} \alpha} + \delta^{\gamma}_{\beta} \w{T}^{\phantom{\rho} \mu}_{\mu \phantom{\mu} \alpha} - \delta^{\gamma}_{\beta} \underbrace{\w{T}^{\phantom{\rho} \mu}_{\alpha \phantom{\mu} \mu}}_{=0} \right) \right. \\
		&+ \left. bc \left( \delta_{\mu}^{\mu} \delta^{\gamma}_{\alpha} \w{T}^{\lambda}_{\phantom{\lambda} \beta \lambda} - \delta_{\mu}^{\mu} \delta^{\gamma}_{\beta} \w{T}^{\lambda}_{\phantom{\lambda} \alpha  \lambda} + \delta_{\beta}^{\mu} \delta^{\gamma}_{\mu} \w{T}^{\lambda}_{\phantom{\lambda} \alpha \lambda} \right. \right.\\
		& \left. \left. \phantom{+ bc \left( \delta_{\mu}^{\mu} \delta^{\gamma}_{\alpha} \w{T}^{\lambda}_{\phantom{\lambda} \beta \lambda} \right) } - \delta_{\beta}^{\mu} \delta^{\gamma}_{\alpha} \w{T}^{\lambda}_{\phantom{\lambda} \mu \lambda} + \delta_{\alpha}^{\mu} \delta^{\gamma}_{\beta} \w{T}^{\lambda}_{\phantom{\lambda} \mu \lambda} - \delta_{\alpha}^{\mu} \delta^{\gamma}_{\mu} \w{T}^{\lambda}_{\phantom{\lambda} \beta \lambda} \right) \right]  \\
		& + \left( a - b \right) c \left( \w{T}^{\phantom{\alpha \beta} \gamma}_{\alpha \beta} - \w{T}^{\phantom{\alpha \beta} \gamma}_{\beta \alpha} + \w{T}^{\gamma}_{\phantom{\gamma} \alpha \beta} - \w{T}^{\phantom{\alpha} \gamma}_{\alpha \phantom{\gamma} \beta} + \w{T}^{\phantom{\beta} \gamma}_{\beta \phantom{\gamma} \alpha} - \w{T}^{\gamma}_{\phantom{\gamma} \beta \alpha} \right), 
	\end{align*}
	
	\begin{equation*}
		\begin{aligned}
			- \left( \ostar \ostar \w{T} \right)^{\gamma}_{\phantom{\gamma} \alpha \beta} &= \left( 4 a^2 \w{T}^{\gamma}_{\phantom{\gamma} \alpha \beta} + 2ab \w{T}^{\phantom{\rho} \gamma}_{\alpha \phantom{\gamma} \beta} - 2ab \w{T}^{\phantom{\rho} \gamma}_{\beta \phantom{\gamma} \alpha} + 2ac \delta_{\alpha}^{\gamma} \w{T}^{\lambda}_{\phantom{\lambda} \lambda \beta} - 2ac \delta_{\beta}^{\gamma} \w{T}^{\lambda}_{\phantom{\lambda} \lambda \alpha} \right) \\
			&+ \left[ 2ab \left( \delta^{\gamma}_{\alpha} \w{T}^{\mu}_{\phantom{\mu} \beta \mu} + \w{T}^{\gamma}_{\phantom{\mu} \alpha \beta} + \delta^{\gamma}_{\beta} \w{T}^{\mu}_{\phantom{\mu} \mu \alpha} \right) \right. \\
			&+ \left. b^2 \left( - \delta^{\gamma}_{\alpha} \w{T}^{\phantom{\rho} \mu}_{\mu \phantom{\mu}  \beta} + \w{T}^{\phantom{\rho} \gamma}_{\alpha \phantom{\mu}  \beta} - \w{T}^{\phantom{\rho} \gamma}_{\beta \phantom{\mu} \alpha} + \delta^{\gamma}_{\beta} \w{T}^{\phantom{\rho} \mu}_{\mu \phantom{\mu} \alpha} \right) \right. \\
			&+ \left. bc \left( 4 \delta^{\gamma}_{\alpha} \w{T}^{\lambda}_{\phantom{\lambda} \beta \lambda} - 4 \delta^{\gamma}_{\beta} \w{T}^{\lambda}_{\phantom{\lambda} \alpha  \lambda} + \delta_{\beta}^{\gamma} \w{T}^{\lambda}_{\phantom{\lambda} \alpha \lambda} - \delta^{\gamma}_{\alpha} \w{T}^{\lambda}_{\phantom{\lambda} \beta \lambda} + \delta^{\gamma}_{\beta} \w{T}^{\lambda}_{\phantom{\lambda} \alpha \lambda} - \delta_{\alpha}^{\gamma} \w{T}^{\lambda}_{\phantom{\lambda} \beta \lambda} \right) \right]  \\
			& + \left( a - b \right) c \left( 2 \w{T}^{\gamma}_{\phantom{\gamma} \alpha \beta} + 2 \w{T}^{\phantom{\alpha \beta} \gamma}_{\alpha \beta} + 2 \w{T}^{\phantom{\alpha} \gamma}_{\beta \phantom {\gamma} \alpha} \right) \\
			&= \left( 4 a^2 \w{T}^{\gamma}_{\phantom{\gamma} \alpha \beta} + 2ab \w{T}^{\phantom{\rho} \gamma}_{\alpha \phantom{\gamma} \beta} - 2ab \w{T}^{\phantom{\rho} \gamma}_{\beta \phantom{\gamma} \alpha} + 2ac \delta_{\alpha}^{\gamma} \w{T}^{\lambda}_{\phantom{\lambda} \lambda \beta} - 2ac \delta_{\beta}^{\gamma} \w{T}^{\lambda}_{\phantom{\lambda} \lambda \alpha} \right) \\
			&+ \left[ 2ab \left( \delta^{\gamma}_{\alpha} \w{T}^{\mu}_{\phantom{\mu} \beta \mu} + \w{T}^{\gamma}_{\phantom{\mu} \alpha \beta} + \delta^{\gamma}_{\beta} \w{T}^{\mu}_{\phantom{\mu} \mu \alpha} \right) \right. \\
			&+ \left. b^2 \left( - \delta^{\gamma}_{\alpha} \w{T}^{\phantom{\rho} \mu}_{\mu \phantom{\mu}  \beta} + \w{T}^{\phantom{\rho} \gamma}_{\alpha \phantom{\mu}  \beta} - \w{T}^{\phantom{\rho} \gamma}_{\beta \phantom{\mu} \alpha} + \delta^{\gamma}_{\beta} \w{T}^{\phantom{\rho} \mu}_{\mu \phantom{\mu} \alpha} \right) \right. \\
			&+ \left. 2bc \left( \delta^{\gamma}_{\alpha} \w{T}^{\lambda}_{\phantom{\lambda} \beta \lambda} - \delta^{\gamma}_{\beta} \w{T}^{\lambda}_{\phantom{\lambda} \alpha  \lambda} \right) \right]  \\
			& + 2 \left( a - b \right) c \left( \w{T}^{\gamma}_{\phantom{\gamma} \alpha \beta} + \w{T}^{\phantom{\alpha \beta} \gamma}_{\alpha \beta} + \w{T}^{\phantom{\beta} \gamma}_{\beta \phantom{\gamma} \alpha} \right) \\
			&= \w{T}^{\gamma}_{\phantom{\gamma} \alpha \beta} \left( 4a^2 + 2ab + 2ac - 2bc \right) + \w{T}^{\phantom{\alpha \beta} \gamma}_{\alpha \beta}\left( -2ab - b^2 + 2ac - 2bc \right) \\
			&+ \w{T}^{\phantom{\beta} \gamma}_{\beta \phantom{\gamma} \alpha} \left( -2ab - b^2 + 2ac - 2bc \right) + \w{T}^{\lambda}_{\phantom{\lambda} \lambda \left[ \alpha \right.} \delta^{\gamma}_{\left. \beta \right]} \left( -4ac + 4ab + 2b^2 + 4bc \right) \\
			&= \w{T}^{\gamma}_{\phantom{\gamma} \alpha \beta} \left( 4a^2 + 2ab + 2ac - 2bc \right) + \w{T}^{\phantom{\left[ \alpha \beta \right]} \gamma}_{\alpha \beta} \left( -4ab - 2b^2 + 4ac - 4bc \right) \\
			&+ \w{T}^{\lambda}_{\phantom{\lambda} \lambda \left[ \alpha \right.} \delta^{\gamma}_{\left. \beta \right]} \left( -4ac + 4ab + 2b^2 + 4bc \right).
		\end{aligned}
	\end{equation*}
	Requiring $\left( \ostar \ostar \w{T} \right)^{\gamma}_{\phantom{\gamma} \alpha \beta} \overset{!}{=} - \w{T}^{\gamma}_{\phantom{\gamma} \alpha \beta}$ gives us two equations
	\begin{align}
		4a^2 + 2ab + 2ac - 2bc &= + 1,    \\
		b^2 + 2ab - 2ac + 2bc &= 0.
	\end{align}
	
	\chapter{Dual connection} \label{append:connection}
	\pagestyle{plain}
	In section \ref{dual_connection} we are attempting to construct $\left( \ostar \w{\omega} \right)$, where $\w{\omega}^a_{\phantom{a} b \mu}$ are connection 1-forms (of teleparallel connection). We have to also consider a 3-form $\w{\Xi}^a_{\phantom{a}b \mu \nu \lambda}$ such that
	\begin{equation}
		\eta_{ac} \w{\Xi}^c_{\phantom{c}b \mu \nu \lambda} = - \eta_{bc} \w{\Xi}^c_{\phantom{c}a \mu \nu \lambda}.
	\end{equation}
	That is because the result of $\left( \ostar \w{\omega} \right)$ is a 3-form with these properties.
	
	We have parametrized $\left( \ostar \w{\omega} \right)$ using three real parameters $a,b,c$ and $\left( \ostar \w{\Xi} \right)$ using six real parameters $\alpha, \beta, \gamma, \mu, \nu, \lambda$. Now we let \textit{Mathematica} calculate $\left( \ostar \ostar \w{\omega} \right)$ and  $\left( \ostar \ostar \w{\Xi} \right)$ for us and require
	\begin{align}
		\left( \ostar \ostar \w{\omega} \right) = \w{\omega}, &&	\left( \ostar \ostar \w{\Xi} \right) = \w{\Xi}.
	\end{align}
	Doing so, the first condition gives us
	\begin{align}
		1 &= a \left( - 4 \mu + 2 \nu - 4 \lambda \right) + b \left( 6 \alpha + 4 \beta - 2\gamma + 4 \mu - 2 \nu + 2 \lambda \right), \\
		0 &= c \left( 6 \alpha + 2\mu \right) - 2\beta \left( a + b + c \right) + 2\gamma \left( a + b + 2c \right) + \nu \left( a + b + 2c \right), \label{redundant_eq}\\
		0 &= a \left( 6 \alpha + 4 \beta - 2 \gamma + 6 \mu - 3\nu + 4 \lambda \right) + b \left( -2 \mu + \nu - 2 \lambda \right), \\
		0 &= a \left( 2 \beta - 2 \gamma - \nu \right) + b \left( 2 \beta - 2 \gamma - \nu  \right) + c \left( -6 \alpha + 2 \beta -4 \gamma -2 \mu -2 \nu \right) \label{last_eq},
	\end{align}
	while the second condition yields
	\begin{align}
		1 &= 12a \alpha + 6b \alpha + c \left( 12 \alpha -12 \beta \right), \\
		0 &= 4a \beta + 2b \beta + c \left( -6 \alpha +10 \beta +4 \gamma \right), \label{eq_I} \\
		0 &= a \left( 6 \alpha + 2\beta - 2\mu \right) -2b \mu + c \left( - 2\mu - 2\nu \right), \label{eq_II} \\
		0 &= 4a \gamma +2b \gamma + c \left( 8 \beta +8 \gamma \right), \label{eq_III} \\
		0 &= a \left( 2 \beta -2 \gamma + \nu \right) +b \nu + c \left( 2 \mu +2 \nu \right), \\
		0 &= a \left( 4 \mu -2 \nu +2 \lambda \right) -b \lambda. \label{eq_V}
	\end{align}
	Rewriting (\ref{redundant_eq})
	\begin{equation}
		0 = a \left( -2\beta + 2\gamma + \nu \right) + b \left( -2\beta + 2\gamma + \nu \right) + c \left( 6\alpha -2\beta +4\gamma + 2\mu +2\nu \right),
	\end{equation}
	we see that it is in fact identical with (\ref{last_eq}), so we have
	\begin{align}
		1 &= a \left( - 4 \mu + 2 \nu - 4 \lambda \right) + b \left( 6 \alpha + 4 \beta - 2\gamma + 4 \mu - 2 \nu + 2 \lambda \right), \label{eq:first}\\
		0 &= a \left( 6 \alpha + 4 \beta - 2 \gamma + 6 \mu - 3\nu + 4 \lambda \right) + b \left( -2 \mu + \nu - 2 \lambda \right), \label{eq_IV} \\
		0 &= a \left( 2 \beta - 2 \gamma - \nu \right) + b \left( 2 \beta - 2 \gamma - \nu  \right) + c \left( -6 \alpha + 2 \beta -4 \gamma -2 \mu -2 \nu \right). 
	\end{align}
	Furthermore, we do not need the last equation, because we can get it by taking the sum of eqs. (\ref{eq_I}), (\ref{eq_II}) and 2 times eq. (\ref{eq_V}), and then subtracting eqs. (\ref{eq_III}) and (\ref{eq_IV}). We thus have three equations with no $c$-term\footnote{We have used eq. (\ref{eq_V}) to simplify eqs. (\ref{eq:first}) and (\ref{eq_IV}).}:
	\begin{align}
		1 &= - 2a\lambda + b \left( 6 \alpha + 4 \beta - 2\gamma + 4 \mu - 2 \nu + \lambda \right), \\
		0 &= a \left( 6 \alpha + 4 \beta - 2 \gamma - 2 \mu + \nu \right) + b \left( -2 \mu + \nu \right), \label{vvv} \\
		0 &= a \left( 4 \mu -2 \nu +2 \lambda \right) -b \lambda,
	\end{align}
	and five containing a $c$-term:
	\begin{align}
		1 &= 12a \alpha + 6b \alpha + c \left( 12 \alpha -12 \beta \right), \label{aaa} \\
		0 &= 4a \beta + 2b \beta + c \left( -6 \alpha +10 \beta +4 \gamma \right), \label{bbb} \\
		0 &= a \left( 6 \alpha + 2\beta - 2\mu \right) -2b \mu + c \left( - 2\mu - 2\nu \right), \label{ccc} \\
		0 &= 4a \gamma +2b \gamma + c \left( 8 \beta +8 \gamma \right), \label{ddd} \\
		0 &= a \left( 2 \beta -2 \gamma + \nu \right) +b \nu + c \left( 2 \mu +2 \nu \right). \label{eee}
	\end{align}
	Notice, however, that adding (\ref{ccc}) and (\ref{eee}) gives us (\ref{vvv}), so we are free to only keep eqs. (\ref{vvv}) and (\ref{ccc}) out of the three. Additionally, adding to (\ref{aaa}) twice the eq. (\ref{bbb}) and subtracting (\ref{ddd}) leads to equation
	\begin{equation}
		1 = \left( 2a + b \right) \left( 6\alpha + 4\beta - 2\gamma \right), 
	\end{equation}
	which also does not contain a $c$-term. Thus we have 4 equations without a $c$-term:
	\begin{align}
		1 &= - 2a\lambda + b \left( 6 \alpha + 4 \beta - 2\gamma + 4 \mu - 2 \nu + \lambda \right), \label{jjj} \\
		1 &= \left( 2a + b \right) \left( 6\alpha + 4\beta - 2\gamma \right), \label{kkk} \\
		0 &= a \left( 6 \alpha + 4 \beta - 2 \gamma - 2 \mu + \nu \right) + b \left( -2 \mu + \nu \right), \label{iii} \\
		0 &= a \left( 4 \mu -2 \nu +2 \lambda \right) -b \lambda \label{zizizi}
	\end{align}
	(only 3 of which are independent -- we actually do not need the first one, as it can be obtained by subtracting eq. (\ref{zizizi}) and twice the eq. (\ref{iii}) from eq. (\ref{kkk})), and 3 equations with a $c$-term:
	\begin{align}
		0 &= 4a \beta + 2b \beta + c \left( -6 \alpha +10 \beta +4 \gamma \right), \\
		0 &= a \left( 6 \alpha + 2\beta - 2\mu \right) -2b \mu + c \left( - 2\mu - 2\nu \right), \\
		0 &= 4a \gamma +2b \gamma + c \left( 8 \beta +8 \gamma \right).
	\end{align}
	The very first equation says that we can not have $a=b=0$ (which means that the other two equations must be linearly dependent), and the second equation tells that $b \neq -2a$. For what follows, it is useful to notice that $b \neq -a$: if $b =-a$ were true, then eq. (\ref{iii}) would reduce to $0 = a \left( 6 \alpha + 4 \beta - 2 \gamma \right)$, which leads to either $a=b=0$, in contradiction with (\ref{jjj}), or $6 \alpha + 4 \beta - 2 \gamma = 0$, contradicting (\ref{kkk}).
\end{appendices}

	\bibliographystyle{Style}
	\bibliography{bibliography}

\end{document}